\documentclass[
  prx, 
  10pt,
  amsmath,amssymb,
  aps,
  twocolumn,
  superscriptaddress,
  floatfix,
  nofootinbib,
]{revtex4-2}

\usepackage{graphicx}
\usepackage[caption=false]{subfig} 
\usepackage{microtype}
\usepackage{bm}
\usepackage{braket}
\usepackage{xcolor}

\usepackage{mathtools}
\usepackage{amsthm}
\usepackage{tikz}
\usepackage{booktabs}
\usepackage{multirow}
\usepackage{hyperref}
\usepackage[capitalise,noabbrev]{cleveref}
\usepackage{refcount} 
\crefname{paragraph}{paragraph}{paragraphs}
\Crefname{paragraph}{Paragraph}{Paragraphs}
\crefname{problem}{Problem}{Problems}
\Crefname{problem}{Problem}{Problems}
\newtheorem{theorem}{Theorem}

\newcounter{savedtheorem}
\newenvironment{restatedtheorem}[1]{%
  \setcounter{savedtheorem}{\value{theorem}}%
  \setcounterref{theorem}{#1}\addtocounter{theorem}{-1}%
  \begin{theorem}}{%
  \end{theorem}\setcounter{theorem}{\value{savedtheorem}}}
\newtheorem{lemma}{Lemma}
\newcounter{savedlemma}
\newenvironment{restatedlemma}[1]{%
  \setcounter{savedlemma}{\value{lemma}}%
  \setcounterref{lemma}{#1}\addtocounter{lemma}{-1}%
  \begin{lemma}}{%
  \end{lemma}\setcounter{lemma}{\value{savedlemma}}}
\theoremstyle{definition}
\newtheorem{problem}{Problem}
\theoremstyle{remark}

\DeclareMathOperator{\Tr}{tr}
\hypersetup{colorlinks=true,linkcolor=blue,citecolor=blue,urlcolor=blue}
\DeclareMathOperator{\tr}{tr}

\newcommand{\SC}{\mathsf{SC}}

\newcommand{\NL}{\mathsf{NL}}
\newcommand{\BPL}{\mathsf{BPL}}

\newcommand{\Lclass}{\mathsf{BQL}}

\begin{document}

\title{Classical and Quantum Simulation of All-to-All Quantum Dynamics in Logarithmic Space}

\author{Maximilian Lutz}
\thanks{These authors contributed equally}
\email{maximilian.lutz@mpq.mpg.de}
\author{Rahul Trivedi}
\thanks{These authors contributed equally}
\email{rahul.trivedi@mpq.mpg.de}
\author{Ignacio Cirac}
\email{ignacio.cirac@mpq.mpg.de}
\affiliation{Max-Planck-Institut f\"ur Quantenoptik, Hans-Kopfermann-Stra\ss e 1, D-85748 Garching, Germany}
\affiliation{Munich Center for Quantum Science and Technology (MCQST), Schellingstra\ss e 4, D-80799 M\"unchen, Germany}

\date{\today}

\begin{abstract}
We study both the classical and the quantum complexity of the problem of estimating disorder-averaged quench dynamics for all-to-all spin Hamiltonians such as quantum Hopfield, Sherrington--Kirkpatrick (SK) or Dicke-type models.
Disorder averaging restores permutation invariance in this setting, so states of $N$ spins can be stored in $\mathcal O(\log N)$ qubits.
While logarithmic-space quantum computations are known to be simulatable classically in polynomial time or in polylogarithmic space, whether a single algorithm can achieve both is open, raising the possibility of \emph{simultaneous} quantum advantage.
To this end, we give a quantum algorithm for permutation-invariant $k$-local dynamics without disorder using polynomial time and $\mathcal O(\log N)$ quantum and classical space.
We show that any classical algorithm matching these resources would imply $\Lclass\subseteq\SC$. For short times, however, we provide such algorithms, including for SK models.
For Hopfield models, we discuss how dilating the disorder to ancilla qubits reduces the problem to the disorder-free case on an enlarged system, thus extending both classical and quantum results to them.
We give some evidence that their long-time dynamics are non-trivial.
They thus emerge as a candidate for a simultaneous time--space advantage in quantum simulation.
\end{abstract}

\maketitle

\section{Introduction}
\label{sec:intro}

Few models generate as much complexity from as little structure as disordered all-to-all spin models.
The Sherrington--Kirkpatrick (SK) model~\cite{sherrington1975}, its $p$-spin generalisations~\cite{gross1984}, and the Hopfield models~\cite{amit1985,ma1993,nishimori1996} couple every spin to every other through random interactions, and from this compact definition emerge replica-symmetry breaking, glassy energy landscapes, and the phenomenology behind disordered magnets, neural networks, and hard optimisation problems~\cite{parisi1979,mezard1987,ma1993,nishimori1996,nishimori2001}.
Since the couplings are random, the meaningful quantities are averages over the disorder - in this work, we focus on such averages for quantum quench dynamics.

Numerically simulating these dynamics faces challenges:
without geometric locality there is no light cone to exploit, and finite-size effects are strong, so larger systems are needed to approach the thermodynamic limit than for lattice models \cite{aspelmeier2008finite}.
The standard many-body toolbox of Krylov-subspace, tensor-network, and Monte Carlo methods therefore applies only in part.

Quantum computers are naturally suited to simulate dynamics of many-body systems~\cite{lloyd1996}, but a direct simulation needs one qubit per spin, while early fault-tolerant devices will offer few logical qubits.
Memory, in addition to running time, is therefore a decisive resource.
Relevant to this, in our setting, disorder averaging restores a symmetry the individual realisations lack: the ensemble is invariant under relabelling the sites, so the averaged dynamics are permutation invariant and live in a subspace whose dimension is only polynomial in the number of sites $N$.
In principle, $\mathcal{O}(\log N)$ qubits therefore suffice to store the state of ensemble averages of all-to-all models on $N$ spins, providing a starting point for space-efficient quantum algorithms.

However, logarithmic quantum space is also a regime with strong classical simulation results.
Any polynomial-time quantum computation on logarithmically many qubits can be simulated classically either in polynomial time, or in space only quadratically larger than the quantum space~\cite{watrous1999,watrous2003,borodin1983}.
While this seems to rule out superpolynomial speedups, neither simulation meets both bounds at once, and it is this \emph{simultaneous} time--space complexity that matters here.
Indeed this corresponds to the open problem of $\Lclass\overset{\scriptscriptstyle ?}{\subseteq}\SC$, i.e.\ whether all problems solvable with bounded error in polynomial circuit depth and logarithmic quantum space are solvable by a classical Turing machine in both polynomial time and polylogarithmic space.
This inclusion is open; a related classical counterpart for nondeterministic logarithmic space is the famous $\NL\overset{\scriptscriptstyle ?}{\subseteq}\SC$~\cite{arora2009}.

We thus first explore the question of whether physically interesting dynamics are accessible to quantum computers in polynomial time and logarithmic space while remaining out of reach for classical computers under the same bounds.
Concretely, we consider $k$-local all-to-all Hamiltonians whose couplings are built from random signs attached to each site, with a tunable number of sign patterns per site (\cref{sec:setup}).
Few patterns give quantum Hopfield-type models~\cite{ma1993,nishimori1996}; many patterns make the couplings Gaussian and recover the transverse-field SK and $p$-spin models.
Exploiting the symmetry is not immediate, since the averaged state is a permutation-invariant \emph{mixed} state outside the symmetric subspace of pure states, and an algorithm must be designed that never leaves the invariant space.
Without disorder, the dynamics are dictated by Dicke-like Hamiltonians: they are collective-spin dynamics confined to the symmetric subspace.

Our results chart the resulting landscape.
A hardness result implies that there must be a simultaneous space-time quantum advantage for the problem of estimating permutation-invariant quench dynamics at polynomial times under certain complexity assumptions.
We show that the case without disorder, i.e.\ Dicke-type models, can be quantum simulated in polynomial time and logarithmic quantum and classical workspace. 
Classically, however, we meet the same resource bounds at short evolution times, in two ways.
Firstly, we provide sampling algorithms up to $t=\mathcal{O}(\log(N)/N)$ through truncated expansions.
Secondly, for expectation value estimation up to $t=\mathcal{O}(1)$ we show a cluster expansion result that also covers the SK-type models \cite{wild2023}.
Furthermore, we show how discrete disorder can be dilated by additional ancilla qubits, which reduces averaged dynamics of Hopfield-type models to the disorder-free case, with corresponding quantum and classical (sampling and expectation value estimation) results corresponding to the above therefore applying.
Finally, we provide some evidence that the averaged dynamics of Hopfield-type models at long evolution times are non-trivial.
They thus emerge as a candidate for a simultaneous time--space quantum advantage.

\section{Summary of Results}
\label{sec:summary_of_results}

\subsection{Setup: Hamiltonian \& Dynamics Estimation}
\label{sec:setup}
\label{sec:tasks}
We consider a system of $N$ qubits with the collective spin operators $m_\alpha = N^{-1}\sum_{i=1}^N \sigma_i^\alpha, \alpha \in \{x, y, z\}$.
Specifically, we focus on ensembles of $k$-local all-to-all Hamiltonians of the form
\begin{gather}
\label{eq:hamiltonian}
H_J = N f(m_x,m_y,m_z) - \frac{J}{N^{k-1}}\sum_{1\le i_1<\cdots<i_k}^NJ_{i_1\cdots i_k}\sigma^z_{i_1}\cdots \sigma^z_{i_k}\\
J_{i_1\cdots i_k}=\frac{1}{\sqrt{r}}\sum_{\ell=1}^{r}v_{i_1\ell}\cdots v_{i_k\ell},
\end{gather}
where $v_{i\ell}$ are i.i.d.\ Rademacher-distributed, i.e.\ $\pm1$ with equal probabilities, and $f$ is a degree-$k$ polynomial.
Since the indices within each tuple are distinct and $v_{i\ell}^2=1$, the couplings have zero mean and are uncorrelated with unit variance,
\begin{equation}
  \mathbb{E}[J_{i_1\cdots i_k}]=0, \qquad
  \mathbb{E}[J_{i_1\cdots i_k}J_{i'_1\cdots i'_k}]=\delta_{i_1i'_1}\cdots\delta_{i_ki'_k},
  \label{eq:coupling-moments}
\end{equation}
for ordered tuples $i_1<\cdots<i_k$ and $i'_1<\cdots<i'_k$.

Different regimes of this family of Hamiltonians $H_J$ correspond to various models of interest in many-body physics, which we now discuss.

First of all, note that as per Schur--Weyl duality, observables on qubits commuting with every site permutation $U_\pi$ are exactly those generated by $\{m_\alpha\}_{\alpha=x,y,z}$, so for $J=0$, $H_J$ corresponds to all permutation-invariant (i.e.\ $U_\pi H U_\pi^\dagger=H$ for all $\pi$) $k$-local qubit Hamiltonians.
Physically, such models closely relate to Dicke-type models: e.g.\ for $k=2$, this class contains the Lipkin--Meshkov--Glick (LMG) model~\cite{lipkin1965}, which arises as the effective spin Hamiltonian of the Dicke model of $N$ two-level atoms collectively coupled to a single cavity mode~\cite{dicke1954,kirton2019} when the cavity mode is far detuned and adiabatically eliminated~\cite{larson2010}.
We later also consider the qudit generalisation of local dimension $\chi$ of this case, where $\{m_\alpha\}_{\alpha=x,y,z}$ is replaced by an operator basis on $\mathbb{C}^\chi$.

Second, we take $J\neq0$ and fixed $r=\mathcal{O}(1)$.
For $k=2$, $H_J$ with $f=Bm_x$ coincides, up to an additive constant, with the quantum Hopfield models~\cite{ma1993,nishimori1996}
\begin{equation}
    H_{\mathrm{Hop}}=BN m_x + \frac1N\sum_{\ell=1}^{r}\mu_\ell A_\ell^2, \qquad
    A_\ell=\sum_i v_{i\ell}\sigma^z_i,
    \label{eq:models-hopfield}
\end{equation}
with uniform pattern weights $\mu_\ell=-J/(2\sqrt r)$.
These models have been studied as a transverse-field extension of the classical neural-network model of associative memory~\cite{amit1985,ma1993,nishimori1996}, where the $r$ sign patterns $(v_{i\ell})_i$ are the stored memories.
We allow general real weights $\mu_\ell$ as mild generalisation of \cref{eq:hamiltonian}.

Finally, we consider the limit of $r \rightarrow \infty$.
For $r\ge N^{k-1}$, we replace the prefactor $J/N^{k-1}$ in \cref{eq:hamiltonian} by $J/N^{(k-1)/2}$ to normalise $\|H_J\|=\mathcal O(N)$ with high probability (see \cref{app:normalization}); this ensures the disordered part does not form a vanishing part of the norm.
Here the couplings of \cref{eq:hamiltonian} become independent Gaussians.
Indeed, writing $J^{(\ell)}_{i_1\cdots i_k}=v_{i_1\ell}\cdots v_{i_k\ell}$ for the rank-one couplings of pattern $\ell$, we have $J_{i_1\cdots i_k}=r^{-1/2}\sum_{\ell=1}^r J^{(\ell)}_{i_1\cdots i_k}$.
The $r$ families $\{J^{(\ell)}_{i_1\cdots i_k}\}_{i_1<\cdots<i_k}$ are i.i.d., since each pattern carries its own signs, and each has zero mean and identity covariance by \cref{eq:coupling-moments} with $r=1$; hence for any finite set of tuples, the multivariate central limit theorem implies convergence to independent standard Gaussians.
Thus $H_J$ in this regime contains the well-known family of quantum $p$-spin Hamiltonians~\cite{gross1984},
\begin{equation}
    H_p= BN m_x+\frac{J}{N^{(p-1)/2}}
    \sum_{i_1<\cdots<i_p}
    g_{i_1\cdots i_p}\sigma^z_{i_1}\cdots\sigma^z_{i_p} 
  \label{eq:models-p-spin}
\end{equation}
with $g_{i_1\cdots i_p} \overset{\mathrm{i.i.d.}}{\sim} \mathcal{N}(0, 1)$ and transverse-field strength $B$.
For $p=2$ and $J=1$, these in turn specialize to the seminal transverse-field Sherrington--Kirkpatrick (SK) model~\cite{sherrington1975},
\begin{equation}
  H_{\mathrm{SK}}=BN m_x + \frac1{\sqrt N}\sum_{i<j}
  g_{ij}\sigma_i^z\sigma_j^z,
  \qquad g_{ij}\overset{\mathrm{i.i.d.}}{\sim}\mathcal N(0,1),
  \label{eq:models-sk}
\end{equation}
which at $B=0$, where all terms commute, reduces to the classical SK Hamiltonian.

We study the computational task of estimating disorder-averaged dynamics from a symmetric product state.
We set $\hbar=1$ throughout.
\begin{problem}[Disorder-averaged dynamics estimation]
\label{prob:averaged-dynamics}
Let $\rho(0)=(\ket{\phi}\bra{\phi})^{\otimes N}$ be a product state, $O$ an observable supported on at most $k_O$ sites, and $t\ge0$ an evolution time.
Given $\epsilon > 0$, compute $\widetilde{F}_O(t)$ with
\begin{gather}
  \overline F_O(t)=\mathbb{E}_J\!\left[\Tr\!\left(O\rho(t)\right)\right],\quad\rho(t)=e^{-iH_Jt}\rho(0)e^{iH_Jt}\\
  \mathrm{s.t.} \, | \widetilde{F}_O(t) - \overline F_O(t) | \le \epsilon \| O\|.
  \label{eq:task-averaged}
\end{gather}
\end{problem}
Throughout, $k$ and $k_O$ are independent of $N$ unless stated otherwise.
Since $\rho(0)$ and the distribution of $H_J$ are invariant under relabelling the sites, so is the averaged state, and only the average of $O$ over site permutations enters; e.g.\ $\sigma^z_1$ and $m_z$ have the same expectation value.
Conversely, $k_O$-local permutation-invariant observables such as $m_z^2$ reduce to $k_O$ observables of this form at a constant loss in accuracy (\cref{app:local-observables}).
We also consider the case where $H_J$ is piecewise-constant time dependent.

\begin{table*}[t]
  \caption{Overview of results.
  Every algorithm listed runs in $\mathrm{poly}(N)$ time and, simultaneously, $\mathcal O(\log N)$ (classical or quantum and classical) space.
  The referenced statements state conditions on accuracy, the locality entries, and the admissible initial states and observables.
  A question mark indicates that no algorithm with these resources is known, the cross marks a hardness result.}
  \label{tab:overview}
  \small
  \renewcommand{\arraystretch}{1.25}
  \begin{tabular*}{\textwidth}{@{\extracolsep{\fill}}llccc@{}}
    \toprule
    & & Sampling & \multicolumn{2}{c}{Expectation value (\cref{prob:averaged-dynamics})} \\
    \cmidrule(lr){3-3}\cmidrule(l){4-5}
    Regime of $H_J$ & Algorithm & $t=\mathcal O(\log N/N)$ & $t=\mathcal O(1),~\epsilon=\Theta(1)$ & $t=\mathrm{poly}(N),~\epsilon=1/\mathrm{poly}(N)$ \\
    \midrule
    $J=0$ (permutation invariant)
      & classical & \cref{thm:informal-sampling} & \cref{thm:informal-classical} & $\times$ (\cref{cor:informal-classical-hardness}) \\
      & quantum   & & \cref{thm:informal-quantum} & \cref{thm:informal-quantum} \\
    \addlinespace
    $r=\mathcal O(1)$ (quantum Hopfield)
      & classical & \cref{thm:hopfield-sampling} & \cref{thm:hopfield-classical} & ? \\
      & quantum   & & \cref{thm:hopfield-quantum} & \cref{thm:hopfield-quantum} \\
    \addlinespace
    $r\to\infty$ ($p$-spin, SK)
      & classical & ? & \cref{thm:sk-classical} & ? \\
    \bottomrule
  \end{tabular*}
\end{table*}
\Cref{tab:overview} summarises our results with respect to this problem across the three regimes, to be discussed in the following.

\subsection{\texorpdfstring{$J=0$}{J=0}: Non-Disordered Case}
\label{sec:summary-nondisordered}

We first consider the regime of $H_{J=0}$, i.e.\ without disorder.
Here no average is needed, and \cref{prob:averaged-dynamics} asks for
\begin{equation}
  F_O(t) = \tr\left(O e^{-iH_{J=0}t} \rho(0) e^{iH_{J=0}t}\right)
  \label{eq:task-fixed}
\end{equation}
instead of $\overline F_O(t)$ in \cref{eq:task-averaged}.

\emph{Quantum algorithms and advantage.}
In \cref{sec:compressed-simulation} we give a quantum algorithm for this task whose quantum and classical workspace are both logarithmic in $N$.
The starting point is that the dynamics never leave the symmetric subspace, whose dimension is $N+1$ for qubits (and $\mathrm{poly}(N)$ for any fixed local dimension), so the evolved state fits on $\mathcal O(\log N)$ qubits.
This alone does not yield a logarithmic-space algorithm, however: the Hamiltonian, the observable and the evolution circuit, even as matrices on the symmetric subspace, are of polynomial size and cannot be stored completely.
We show these can, however, be generated on the fly by a classical controller confined to logarithmic space.
\begin{theorem}[Informal, logarithmic-space quantum simulation]\label{thm:informal-quantum}
For local dimension $\chi$, \cref{prob:averaged-dynamics}, possibly with piecewise-constant time-dependent $H_{J=0}$, can be solved to error $\epsilon$ by a bounded-error quantum algorithm using, simultaneously, $\mathrm{poly}(t,1/\epsilon)\,N^{\mathcal O(\chi)}$ time, $\mathcal O(\chi\log N)$ qubits and $\mathcal O(\chi\log N+\log(t+1/\epsilon))$ bits of classical workspace.
\end{theorem}

We briefly recall some relevant complexity theory background~\cite{arora2009,watrous1999} before we proceed.
In the sublinear-space settings considered here, the input tape, holding an input of size $n$, is assumed to be read-only and does not count towards the space bound, whereas all working memory does.
Importantly, for a quantum algorithm all classical workspace of the circuit generation is counted as well.
$\Lclass$ is the class of promise problems solvable by a quantum algorithm that uses polynomial time and $\mathcal O(\log n)$ qubits, succeeds with probability at least $2/3$, and has its gates emitted by a classical controller of $\mathcal O(\log n)$ workspace.
Its classical counterparts are $\BPL$, decided with bounded error by randomised machines in polynomial time and $\mathcal O(\log n)$ space, and $\SC$, decided by deterministic machines in polynomial time and, simultaneously, $\mathrm{polylog}(n)$ space.
$\Lclass$ is classically simulable in polynomial space and polynomial time by explicit state-vector evolution, or, via \citet{watrous1999}, in $\mathcal O(\log^2 n)$ space and quasipolynomial time~\cite{watrous2003,borodin1983}.
However, the question of whether $\Lclass\overset{\scriptscriptstyle ?}{\subseteq}\SC$ (let alone $\Lclass\overset{\scriptscriptstyle ?}{\subseteq}\BPL$, as $\BPL\subseteq\SC$~\cite{nisan1994}) i.e.\ whether a simultaneous polynomial time, polylogarithmic space simulation is possible, is open.
There is evidence against such a simulation: approximating an entry of the inverse of a well-conditioned matrix of polynomial dimension to inverse-polynomial precision is $\Lclass$-complete~\cite{tashma2013,fefferman2018,fefferman2021}, yet no deterministic or randomised classical algorithm for this task using $o(\log^2 n)$ space is known~\cite{tashma2013}.

We also show a complementary result to \cref{thm:informal-quantum}: simulating permutation-invariant dynamics is at least as hard as running an arbitrary quantum circuit on logarithmically many qubits under reductions that use polynomial time and logarithmic space.
The idea is to encode the circuit's qubits into Dicke states of the symmetric subspace and to compile any circuit on them into a piecewise-constant schedule of a two-local permutation-invariant Hamiltonian.
That such schedules exist follows from the fact that the collective operators $m_x$, $m_y$, $m_z$ and $m_z^2$ generate the full Lie algebra $\mathfrak{su}(N+1)$~\cite{giorda2003}.
This controllability argument is not constructive, however, while a hardness result for $\Lclass$ requires a reduction that itself runs in polynomial time and logarithmic classical space.
In \cref{sec:hardness} we therefore give an explicit compilation fulfilling these requirements.
From this we derive a simultaneous space-time quantum advantage for \cref{prob:averaged-dynamics} under complexity assumptions.


\begin{theorem}[Informal, simultaneous space-time quantum advantage up to complexity assumption]\label{cor:informal-classical-hardness}
Suppose a deterministic classical algorithm solves \cref{prob:averaged-dynamics} for piecewise-constant $H_{J=0}$ and some constant accuracy $\epsilon<1/3$, in $\mathrm{poly}(N,t)$ time and, simultaneously, $\mathrm{polylog}(N,t)$ space. Then $\Lclass\subseteq\SC$.
If the algorithm is instead randomised with bounded error and uses $\mathcal O(\log(Nt))$ space, then even $\Lclass=\BPL$.
\end{theorem}

\emph{Classical algorithms.}
So far we considered $t=\mathrm{poly}(N)$, but at short times the picture for classical algorithms changes. 

In \cref{sec:classical-sampling} we consider the task of sampling, that is, drawing a string from the computational-basis measurement distribution of the evolved state, and show that for $t=\mathcal O(\log N/N)$, it is solved by a classical algorithm in polynomial time and logarithmic space.
At these times the evolution operator is approximated by $\mathcal O(\log N)$ terms of its exponential series, the key observation is then that the Hamiltonian in the symmetric subspace is a banded matrix; a matrix element of the truncated series is therefore a sum over $N^{\mathcal O(\chi)}$ paths, each of which can be enumerated and evaluated in logarithmic space.
\begin{theorem}[Informal, classical sampling at short times]\label{thm:informal-sampling}
For local dimension $\chi$ and $t=\mathcal O(\log N/N)$, the measurement distribution of $\rho(t)$ can be sampled to total variation distance $\epsilon$ by a randomised classical algorithm in $(N/\epsilon)^{\mathcal O(\chi)}$ time and, simultaneously, $\mathcal O(\chi\log(N/\epsilon))$ workspace.
\end{theorem}

Moreover, expectation values of local observables can still be estimated classically in logarithmic space and polynomial time for $t=\mathcal O(1)$, although only to constant accuracy, as we show in \cref{sec:cluster-method}.
The method is a cluster expansion \cite{wild2023}, uses neither the symmetric subspace nor permutation invariance and applies to any all-to-all $k$-body Hamiltonian with bounded, suitably normalised terms, in particular to $H_{J=0}$.
\begin{theorem}[Informal, classical simulation at constant times]\label{thm:informal-classical}
For local dimension $\chi$ \cref{prob:averaged-dynamics} can be solved for $t=\mathcal O(1)$ to constant error $\epsilon$ by a deterministic classical algorithm in $\mathrm{poly}(\chi N)$ time and, simultaneously, $\mathcal O(\log(\chi N))$ workspace.
\end{theorem}

\subsection{\texorpdfstring{$r = \mathcal{O}(1)$}{r = O(1)}: Quantum Hopfield Models}
\label{sec:summary-hopfield}

Next, we consider the regime where $J\neq0$ and fixed $r = \mathcal{O}(1)$.

As we show in \cref{sec:reduction-fixed-rank} in this setting \cref{prob:averaged-dynamics} is exactly reproduced by a permutation-invariant, non-disordered Hamiltonian on an enlarged system.
The idea is to promote the disorder to a quantum degree of freedom by associating each site with ancillas and replacing the random disorder by quantum operations on them.
As the resulting Hamiltonian is diagonal in the computational basis of the ancillas, preparing the ancillas in $\ket{+}$, then tracing over them, reproduces the uniform average over sign configurations exactly, as each ancilla basis state selects one sign configuration. 

\begin{lemma}[Informal, reduction of fixed-rank disorder]\label{thm:informal-reduction}
For $H_J$, let $H^\prime_{J=0}$ be obtained by replacing every $v_{i\ell}$ with $\sigma^z_{i\ell}$ on an ancilla qubit; it is a non-disordered, $k$-local, permutation-invariant Hamiltonian on $N$ qudits of local dimension $\chi=2^{r+1}$.
Then \cref{prob:averaged-dynamics} for $H_J$, $\rho(0)$ is identical to \cref{prob:averaged-dynamics} for $H^\prime_{J=0}$, $\rho^\prime(0) = \rho(0)\otimes\ket{+}\bra{+}^{\otimes rN}$.
\end{lemma}

This then allows to connect to the results of \cref{sec:summary-nondisordered} -- in particular, $H^\prime_{J=0}$ corresponds to exactly the setting of \cref{thm:informal-quantum,thm:informal-sampling,thm:informal-classical}.
The corresponding quantum and classical algorithmic results are thus implied for the $r=\mathcal O (1)$ regime with $\chi = 2^{r+1}$, and in particular for the quantum Hopfield models, as in the following.
\begin{theorem}[Informal, quantum simulation of quantum Hopfield models]\label{thm:hopfield-quantum}
For $|B|,\max_\ell|\mu_\ell|\le\mathrm{poly}(N)$ \cref{prob:averaged-dynamics} can be solved to error $\epsilon$ by a bounded-error quantum algorithm using, simultaneously, $\mathrm{poly}(N,t,1/\epsilon)$ time, $\mathcal O(\log N)$ qubits and $\mathcal O(\log N+\log(t+1/\epsilon))$ bits of classical workspace.
\end{theorem}
\begin{theorem}[Informal, classical sampling of quantum Hopfield models at short times]\label{thm:hopfield-sampling}
For $|B|,\max_\ell|\mu_\ell|=\mathcal O(1)$ and $t=\mathcal O(\log N/N)$ the computational-basis measurement distribution of the disorder-averaged state $\mathbb E_J[\rho(t)]$ of \cref{prob:averaged-dynamics} can be sampled to total variation distance $\epsilon$ by a randomised classical algorithm in $\mathrm{poly}(N/\epsilon)$ time and, simultaneously, $\mathcal O(\log(N/\epsilon))$ workspace.
\end{theorem}
In this case, the ancilla bits of each sampled symbol are discarded, as tracing out the ancillas yields $\mathbb E_J[\rho(t)]$ (\cref{sec:reduction-fixed-rank}).
\begin{theorem}[Informal, classical simulation of quantum Hopfield models at constant times]\label{thm:hopfield-classical}
For $|B|,\max_\ell|\mu_\ell|=\mathcal O(1)$ and $t=\mathcal O(1)$ \cref{prob:averaged-dynamics} can be solved to constant error $\epsilon$ by a deterministic classical algorithm in $\mathrm{poly}(N)$ time and, simultaneously, $\mathcal O(\log N)$ workspace.
\end{theorem}

In \cref{app:rank-one-gauge} we show how the dynamics of a Hopfield model can be related to LMG dynamics through a gauge transformation for the special case of $r=1$, which gives some evidence that the long-time dynamics for this model class can be non-trivial.
Moreover, we use the freedom of non-uniform $\mu_\ell$ described above to match the low-order bond statistics of the SK model at fixed $r$, and numerically demonstrate how the Hopfield models recover the SK model for large $r$ (see \cref{sec:numerics}).

This motivates simulating the ensemble-averaged dynamics of these models at polynomial evolution times and inverse-polynomial precision to explore their behaviour.
As the reduction to permutation-invariant qudit systems allows for quantum simulation of this regime in polynomial time with $\mathcal O(\log N)$ qubits, whereas no classical algorithm meeting the same time and space bounds simultaneously is known, and the short-time classical algorithms no longer applying, such fixed-rank Hopfield dynamics could thus potentially provide a candidate setting for a simultaneous time--space quantum advantage.

\subsection{\texorpdfstring{$r \rightarrow \infty$}{r to infinity}: \texorpdfstring{$p$}{p}-Spin \& Sherrington--Kirkpatrick Models}
\label{sec:summary-gaussian}

Finally, we consider the regime of $H_J$ with $J\neq0$, but in the limit of $r \rightarrow \infty$.

As we show in \cref{sec:cluster-method} the constant time result via cluster expansions also extends to this regime with independent Gaussian couplings.

\begin{theorem}[Informal, classical simulation of $p$-spin and SK dynamics at constant times]\label{thm:sk-classical}
For the $p$-spin models with $p,B,J = \mathcal{O}(1)$, and in particular the SK model, \cref{prob:averaged-dynamics} can be solved for $t=\mathcal O(1)$ to constant error $\epsilon$ by a classical algorithm in $\mathrm{poly}(N)$ time and, simultaneously, $\mathcal O(\log N)$ workspace.
\end{theorem}

\section{Technical Discussion}
\label{sec:technical-discussion}

\subsection{Compressed Encoding \& Logarithmic-Space Element Access}
\label{sec:compressed-encoding}
Before we move to the results, we first discuss some recurring notions about encoding the symmetric subspace in qubits and thus expressing permutation-invariant time evolution in a `compressed encoding', as well as how matrix entries can be accessed in this setting in logarithmic space and polynomial time.
\subsubsection{Compressed Encoding}
Consider the setting of a permutation symmetric subspace of $N$ sites with local dimension $\chi$.
For occupations $\mathbf n=(n_0,\ldots,n_{\chi-1})$ with $\sum_an_a=N$, let $\mathcal X_{\mathbf n}$ be the set of basis strings with $n_a$ sites in state $a$, and define
\begin{align}
	M_N(\mathbf n)=\frac{N!}{\prod_an_a!},\qquad
	\ket{\mathbf n}_N=\frac1{\sqrt{M_N(\mathbf n)}}\sum_{x\in\mathcal X_{\mathbf n}}\ket x.
  \label{eq:generalised-dicke-states}
\end{align}
These states form an orthonormal basis of the symmetric subspace.
We encode $\mathbf n$ by storing $n_1,\ldots,n_{\chi-1}$ in binary, each in $\lceil\log_2(N+1)\rceil$ qubits, with $n_0=N-\sum_{a\ge1}n_a$ implicit; within each block, qubit $i$ holds bit $i$ of the binary number, bit $0$ being the least significant.
This uses $q=(\chi-1)\lceil\log_2(N+1)\rceil=\mathcal O(\log N)$ qubits, i.e. a register of dimension $Q=2^q\le(2N+2)^{\chi-1}=\mathrm{poly}(N)$.
Of the $Q$ computational-basis labels of this register, those with $\sum_{a\ge1}n_a\le N$ are valid occupations; the remaining labels are unused.
Let $V$ be the partial isometry mapping each valid label $\mathbf n$ to $\ket{\mathbf n}_N$ and annihilating the unused labels.
Throughout, $\mathsf X_i,\mathsf Y_i,\mathsf Z_i$ denote the Pauli operators on register qubit $i$, distinguished from the physical-spin operators $\sigma_i^\alpha$, and $\mathcal P_q=\{I,\mathsf X,\mathsf Y,\mathsf Z\}^{\otimes q}$ is the set of $Q^2$ Pauli strings on the register.
For illustration, for qubits with $\chi=2$ the occupation $\mathbf n=(N-n,n)$ is labelled by the integer $0\le n\le N$, and the encoded states are the Dicke states
\begin{align}
	\ket n_N=\binom Nn^{-1/2}\sum_{|x|=n}\ket x,\qquad V\ket n=\ket n_N.
	\label{eq:compression-dicke}
\end{align}
A permutation-invariant Hamiltonian $H$ preserves the symmetric subspace, onto which $VV^\dagger$ projects, so the compressed matrix $h=V^\dagger HV$ satisfies
\begin{equation}
  HV=VV^\dagger HV=Vh,\qquad e^{-iHt}V=Ve^{-iht}.
  \label{eq:compression-intertwining}
\end{equation}
Defining the compressed state and observable
\begin{equation}
  \varrho(t) = V^\dagger \rho(t) V,\quad o = V^\dagger O V,
  \label{eq:compression-compressed-state}
\end{equation}
and assuming that $\rho(0)$ is supported on the symmetric subspace, so that $\rho(0)=V\varrho(0)V^\dagger$, \cref{eq:compression-intertwining} gives $\rho(t)=V\varrho(t)V^\dagger$ with $\varrho(t)=e^{-iht}\varrho(0)e^{iht}$, and \cref{eq:task-fixed} can be reformulated, for any observable $O$, to
\begin{equation}
F_O(t) = \tr\left(o e^{-iht} \varrho(0) e^{iht}\right) = \tr\left(o \varrho(t)\right).
\end{equation}

\subsubsection{Element Access to Permutation-Invariant Operators in Logarithmic Space}
\label{sec:element-access-PI-observable}

Let $A$ be a permutation-invariant, $k$-local Hermitian operator on $N$ sites of local dimension $\chi$, and let $a=V^\dagger AV$ be its compressed matrix, with entries $a_{\mathbf m\mathbf n}=\bra{\mathbf m}_NA\ket{\mathbf n}_N$ - we discuss what resources are necessary to compute any such entry.

For a single-site operator $B$ on $\mathbb C^\chi$, let
$
	m(B)=\frac1N\sum_{i=1}^NB_i
$
be the corresponding collective operator.
By Schur--Weyl duality, the operators commuting with all site permutations are exactly the polynomials in collective operators, and $k$-locality restricts the degree to at most $k$.
Hence $A=Nf$ for a polynomial $f$ of degree at most $k$ in collective operators, and we assume that it is supplied in this form on the read-only input tape, in the sense of \cref{sec:summary-nondisordered}, as a list of monomials
\begin{equation}
	f =\sum_\mu c_\mu\,m(B_{\mu,1})\cdots m(B_{\mu,\ell_\mu})
	\label{eq:compression-normalized-family}
\end{equation}
with $\ell_\mu \le k$ and $B_{\mu,j}\in\mathcal B$ 
where $\mathcal B$ is a fixed set of at most $\chi^2$ single-site operators of norm at most one, and $\mu$ runs over distinct words in $\mathcal B$, each listed by the indices of its factors in $\mathcal B$ together with its coefficient $c_\mu$.
We assume all numbers are given as $\mathcal O(\log N)$-bit numbers.
Since the words are distinct, the list has at most $\sum_{\ell=0}^k\chi^{2\ell}\le2\chi^{2k}$ entries.
We write $\|f\|_1=\sum_\mu|c_\mu|$, so that the triangle inequality gives $\|a\|\le\|A\|\le N\|f\|_1$.
Collective operators do not commute, the order of the factors is thus part of the input; and only the sum, not the individual monomials, needs to be Hermitian.

\begin{lemma}[Element access to permutation-invariant operators]
	\label{lem:compression-element-access}
	Let $A=Nf$ be supplied as in \cref{eq:compression-normalized-family}.
	Then any entry $a_{\mathbf m\mathbf n}$ can be computed to error $\epsilon$ in time $\chi^{4k}\,\mathrm{poly}\big(\chi,k,\log(N\|f\|_1/\epsilon)\big)$ and workspace $\mathcal O\big(\chi\log N+k\log\chi+\log(k\|f\|_1/\epsilon)\big)$.
	In particular, for $\chi=\mathcal O(1)$, $k=\mathcal O(\log N)$, and $\|f\|_1,\epsilon^{-1}\le\mathrm{poly}(N)$, this is polynomial time and $\mathcal O(\log N)$ workspace.
\end{lemma}

The proof rests on the action of a collective operator on the occupation basis. With $\mathbf e_a$ the unit occupation vector,
\begin{align}
	m(B)\ket{\mathbf n}_N & =\sum_{a\ne b}B_{ab}\frac{\sqrt{(n_a+1)n_b}}N\ket{\mathbf n+\mathbf e_a-\mathbf e_b}_N\nonumber\\
	                      & \quad+\sum_aB_{aa}\frac{n_a}N\ket{\mathbf n}_N,
	\label{eq:compression-matrix-units}
\end{align}
where impossible occupations give zero.
Indeed, $\sum_iB_i$ changes the state of one site at a time, and its component taking a site from state $b$ to state $a\ne b$, with amplitude $B_{ab}$, maps each string in $\mathcal X_{\mathbf n}$ to $n_b$ strings in $\mathcal X_{\mathbf n'}$, $\mathbf n'=\mathbf n+\mathbf e_a-\mathbf e_b$, each string in $\mathcal X_{\mathbf n'}$ arising $n_a+1$ times; together with $M_N(\mathbf n')/M_N(\mathbf n)=n_b/(n_a+1)$ this gives the square root, while the diagonal part counts the $n_a$ sites in state $a$.
For qubits, with the Dicke states of \cref{eq:compression-dicke} labelled by $n=n_1$, $b_n=\sqrt{(n+1)(N-n)}/N$ and $b_{-1}=b_N=0$, this gives
\begin{align}
	m_x\ket n_N & =b_n\ket{n+1}_N+b_{n-1}\ket{n-1}_N,\nonumber   \\
	m_y\ket n_N & =ib_n\ket{n+1}_N-ib_{n-1}\ket{n-1}_N,\nonumber \\
	m_z\ket n_N & =(1-2n/N)\ket n_N.
	\label{eq:compression-normalized-entries}
\end{align}

\begin{proof}[Proof of \cref{lem:compression-element-access}]
	By \cref{eq:compression-matrix-units}, expanding every factor of a monomial $\mu$ turns $\bra{\mathbf m}_Nm(B_{\mu,1})\cdots m(B_{\mu,\ell_\mu})\ket{\mathbf n}_N$ into a sum over at most $\chi^{2\ell_\mu}$ paths, each a sequence of index pairs $(a_j,b_j)$ selected by a counter of $2\ell_\mu\lceil\log_2\chi\rceil$ bits.
	Applying the factors from right to left, each step updates the occupation and multiplies one term of \cref{eq:compression-matrix-units}, and the path contributes the resulting product $s$ if it ends at $\mathbf m$ and zero otherwise.
	Hence $a_{\mathbf m\mathbf n}=N\sum_\mu c_\mu\sum_{\text{paths}}s$.
	The double loop reads the monomials from the input tape and stores only its position on the tape in $\mathcal O(k\log\chi+\log N)$ bits, the path counter in $\mathcal O(k\log\chi)$ bits, the current occupation in $(\chi-1)\lceil\log_2(N+1)\rceil$ bits, and a constant number of accumulators.
	Every factor of $s$ is a term of \cref{eq:compression-matrix-units}, a product of an entry of $B$, of modulus at most $\|B\|\le1$, and a weight of modulus at most one, so approximating each to error $\delta$ and clipping to the unit disc changes $s$ by at most $\ell_\mu\delta\le k\delta$.
	As the coefficients are read exactly, the computed entry deviates by at most
	\begin{align}
		N\sum_\mu|c_\mu|\,\chi^{2k}k\delta=Nk\chi^{2k}\|f\|_1\delta.
		\label{eq:compression-element-error}
	\end{align}
	The choice $\delta=\epsilon/(Nk\chi^{2k}\|f\|_1)$ makes this at most $\epsilon$ and requires fixed-point arithmetic on $\mathcal O\big(\log(Nk\|f\|_1/\epsilon)+k\log\chi\big)$ bits.
  Each of the at most $2\chi^{2k}$ monomials thus costs at most $\chi^{2k}$ paths, i.e. passes over its entry on the tape, each with $\ell_\mu\le k$ look-ups in $\mathcal B$ and arithmetic operations on numbers of this length, which gives the stated time, and together with the loop state the stated workspace.
\end{proof}

\subsubsection{Element Access to Local Observables in Logarithmic Space}
\label{sec:element-access-local-observable}

The observable $O$ of \cref{prob:averaged-dynamics} is supported on at most $k_O$ sites and in general not permutation invariant, but only its compressed matrix $o=V^\dagger OV$ of \cref{eq:compression-compressed-state} enters.
As the states $\ket{\mathbf n}_N$ are invariant under site permutations, $U_\pi V=V$, so $o$ is unchanged if $O$ is replaced by any relabelled copy $U_\pi^\dagger OU_\pi$; for instance, $\sigma^z_1$ and $m_z$ have the same compressed matrix.
We may therefore take $O=O_{\mathrm{loc}}\otimes I$ with $O_{\mathrm{loc}}$ acting on the first $k_O$ sites and assume that $O_{\mathrm{loc}}$ is supplied on the read-only input tape as its $\chi^{k_O}\times\chi^{k_O}$ matrix with $\mathcal O(\log N)$-bit entries.

\begin{lemma}[Element access to local observables]
	\label{lem:local-element-access}
	Let $O$ be supplied as above.
	Then any entry $o_{\mathbf m\mathbf n}$ can be computed to error $\epsilon$ in time $\chi^{2k_O}\,\mathrm{poly}\big(\chi,k_O,\log(N/\epsilon)\big)$ and workspace $\mathcal O\big(\chi\log N+k_O\log\chi+\log(1/\epsilon)\big)$.
	In particular, for $\chi=\mathcal O(1)$, $k_O=\mathcal O(\log N)$, and $\epsilon^{-1}\le\mathrm{poly}(N)$, this is polynomial time and $\mathcal O(\log N)$ workspace.
\end{lemma}

\begin{proof}
	Collecting the strings in \cref{eq:generalised-dicke-states} by their first $k_O$ sites $x$ gives
	\begin{align}
		\ket{\mathbf n}_N & =\sum_xw_{\mathbf n}(x)\ket x\otimes\ket{\mathbf n-\mathbf c(x)}_{N-k_O},\nonumber \\
		w_{\mathbf n}(x)^2 & =\frac{M_{N-k_O}(\mathbf n-\mathbf c(x))}{M_N(\mathbf n)},
		\label{eq:local-weights}
	\end{align}
	where $x$ runs over the $\chi^{k_O}$ basis states of $k_O$ sites and $\mathbf c(x)$ is the occupation of $x$.
	The ratio $w_{\mathbf n}(x)^2$ is the probability that the first $k_O$ sites of a random string in $\mathcal X_{\mathbf n}$ read $x$, a product of $k_O$ factors (remaining sites in state $x_j$)/(remaining sites), each at most one.
	By orthonormality of the $\ket{\cdot}_{N-k_O}$,
	\begin{equation}
		o_{\mathbf m\mathbf n}=\sum_{x,y}w_{\mathbf m}(y)\,(O_{\mathrm{loc}})_{yx}\,w_{\mathbf n}(x),
		\label{eq:local-entries}
	\end{equation}
	where only pairs with $\mathbf m-\mathbf c(y)=\mathbf n-\mathbf c(x)$ contribute.
	These are at most $\chi^{2k_O}$ terms, each an entry of $O_{\mathrm{loc}}$ times two weights in $[0,1]$; let $\|O_{\mathrm{loc}}\|_{\max}\le\mathrm{poly}(N)$ be the largest modulus of these entries.
	Summing the terms one by one as in the proof of \cref{lem:compression-element-access}, with every weight computed to precision $\epsilon/(2\chi^{2k_O}\|O_{\mathrm{loc}}\|_{\max})$ and clipped to $[0,1]$, changes each term by at most $\epsilon/\chi^{2k_O}$ and hence the entry by at most $\epsilon$.
	The loop stores the pair $(x,y)$ in $2k_O\lceil\log_2\chi\rceil$ bits, the occupations in $\mathcal O(\chi\log N)$ bits, and a constant number of fixed-point numbers of $\mathcal O(k_O\log\chi+\log(N/\epsilon))$ bits, and spends $\mathrm{poly}(\chi,k_O,\log(N/\epsilon))$ operations per term, which gives the stated time and workspace.
\end{proof}

\subsection{Compressed Quantum Simulation}
\label{sec:compressed-simulation}
We now consider the problem of solving the dynamics estimation problem of \cref{eq:task-fixed} for a permutation-invariant Hamiltonian on $N$ sites of local dimension $\chi$ from a symmetric product state $\ket{\psi_0}=\ket\phi^{\otimes N}$ via a quantum algorithm that requires, simultaneously, $N^{\mathcal O(\chi)}$ time, only $\mathcal O(\chi\log N)$ qubits and $\mathcal O(\chi\log N)$ bits of classical workspace.

The starting point is that the dynamics never leave the symmetric subspace, whose dimension is at most $(N+1)^{\chi-1}$, so the evolved state can be described with $\mathcal O(\chi\log N)$ qubits.
For a full construction, care must be taken to only use logarithmic classical space to emit the gates, however.

Every Hermitian $Q\times Q$ matrix $a$ expands in the Pauli strings $\mathcal P_q$ of \cref{sec:compressed-encoding} as
\begin{equation}
	a=\sum_{P\in\mathcal P_q}\alpha_PP,
  ~
	\alpha_P=\frac{\tr(Pa)}{Q}\in\mathbb R,
  ~
	\sum_P\alpha_P^2 =\frac{\tr(a^2)}{Q}\le\|a\|^2.
	\label{eq:compression-pauli}
\end{equation}
The compressed matrices $h$ and $o$ are of polynomial size and cannot be stored, but their Pauli coefficients can be computed on the fly from their entries.

\begin{lemma}[Pauli expansion of compressed operators]
	\label{lem:compression-pauli}
	Let $A$ be a Hermitian operator on $N$ sites of local dimension $\chi$, and let its compressed matrix $a=V^\dagger AV$ have entries $a_{\mathbf m\mathbf n}$ computable to error $\epsilon$ in time $T$ and workspace $W$.
	Then each Pauli coefficient $\alpha_P$ of $a$ in \cref{eq:compression-pauli} can be computed to error $\epsilon$ in time $N^{\mathcal O(\chi)}\big(T+\log(\|a\|/\epsilon)\big)$ and workspace $W+\mathcal O\big(\chi\log N+\log(\|a\|/\epsilon)\big)$, and $\sum_P|\alpha_P|\le Q\|a\|$.
\end{lemma}

\begin{proof}
	Entries of $a$ indexed by an unused register label vanish, since $V$ annihilates such labels; all other entries are computed as assumed.
	Each row of a Pauli string $P$ has a single nonzero entry, a phase in $\{\pm1,\pm i\}$: for a register label $\mathbf n$, the entry $\bra{\mathbf n}P\ket{\mathbf n'}$ is nonzero for exactly one label $\mathbf n'$, obtained from $\mathbf n$ by flipping the bits on which $P$ acts as $\mathsf X$ or $\mathsf Y$, and its value is the product of the corresponding single-qubit phases.
	Hence $Q\alpha_P=\tr(Pa)$ is a sum of $Q$ entries of $a$, each multiplied by such a phase, and approximating every entry to error $\epsilon$ changes $\alpha_P$ by at most $\epsilon$, so the precision of the entries carries over to $\alpha_P$.
	The expansion itself is never stored: the strings $P$ are enumerated one at a time, each $\alpha_P$ is accumulated as a running sum over the $Q$ entries of $a$, and only the current $P$, the entry index $\mathbf n$, the running sum, and the workspace of the entry computation are kept; the first two take $\mathcal O(\chi\log N)$ bits and the running sum, of modulus at most $Q\|a\|$, $\mathcal O(\chi\log N+\log(\|a\|/\epsilon))$ bits.
	The $Q\le(2N+2)^{\chi-1}=N^{\mathcal O(\chi)}$ entries per coefficient, each followed by one addition to the running sum, give the stated time.
	Finally, by the Cauchy--Schwarz inequality and \cref{eq:compression-pauli}, $\sum_P|\alpha_P|\le Q\big(\sum_P\alpha_P^2\big)^{1/2}\le Q\|a\|$.
\end{proof}

For a permutation-invariant Hamiltonian $H=Nf$ supplied as in \cref{lem:compression-element-access}, that lemma provides the entries of $h$, and $\|h\|\le N\|f\|_1$; for a local observable $O$, \cref{lem:local-element-access} provides those of $o$.

\begin{lemma}[Simulation of permutation-invariant Hamiltonians]
	\label{prop:compression-simulation}
	Let $H=Nf$ be a permutation-invariant $k$-local Hamiltonian of local dimension $\chi$ supplied as in \cref{lem:compression-element-access}, let $t>0$ be a logarithmic-bit rational, and let $\epsilon>0$.
	Then there is a uniform circuit $U_{\mathrm{sim}}$ over a fixed finite universal gate set, generated in time $\mathrm{poly}(t,\|f\|_1,1/\epsilon)\,N^{\mathcal O(\chi)}\chi^{\mathcal O(k)}$ and $\mathcal O\big(\chi\log N+k\log\chi+\log(t+\|f\|_1+1/\epsilon)\big)$ classical workspace, acting on the $q$ register qubits alone, such that for some phase $\theta$
	\begin{align}
		\left\|U_{\mathrm{sim}}-e^{i\theta}e^{-iht}\right\|\le\epsilon.
	\end{align}
\end{lemma}

\begin{proof}[Proof of \cref{prop:compression-simulation}]
	By \cref{lem:compression-pauli}, with the entries of $h$ computed by \cref{lem:compression-element-access} and $\|h\|\le N\|f\|_1$, each Pauli coefficient $\alpha_P$ of $h$ is computable to error $\delta$ in time $N^{\mathcal O(\chi)}\chi^{4k}\,\mathrm{poly}(\chi,k,\log(N\|f\|_1/\delta))$ and workspace $\mathcal O(\chi\log N+k\log\chi+\log(k\|f\|_1/\delta))$, with $\sum_P|\alpha_P|\le QN\|f\|_1$.
	Round each coefficient to $\widetilde\alpha_P$ with error at most $\delta=\epsilon/(3Q^2t)$ and set $\widetilde h=\sum_P\widetilde\alpha_PP$. Then $\|h-\widetilde h\|\le Q^2\delta$, so $\|e^{-iht}-e^{-i\widetilde ht}\|\le tQ^2\delta=\epsilon/3$, and $t\sum_P|\widetilde\alpha_P|\le tQN\|f\|_1+\epsilon/3$.
	Split $t$ into $n$ steps of length $v=t/n$. For Hermitian $A_1,\ldots,A_M$, comparing $e^{-iv\sum_\ell A_\ell}$ with $\prod_\ell e^{-ivA_\ell}$ one factor at a time gives the first-order product formula
	\begin{align}
	& \Big\|e^{-iv\sum_\ell A_\ell}-\prod_\ell e^{-ivA_\ell}\Big\|\nonumber\\
	& \quad\le v^2\sum_{\ell<\ell'}\|A_\ell\|\,\|A_{\ell'}\|
	\le\frac{v^2}2\Big(\sum_\ell\|A_\ell\|\Big)^2.
	\label{eq:compression-trotter}
	\end{align}
	Applied to the terms $\widetilde\alpha_PP$ of $\widetilde h$, the error per step is at most $v^2\big(\sum_P|\widetilde\alpha_P|\big)^2/2$, so $n=\lceil3(tQN\|f\|_1+\epsilon/3)^2/(2\epsilon)\rceil$ steps give total product error at most $t^2\big(\sum_P|\widetilde\alpha_P|\big)^2/(2n)\le\epsilon/3$.
	Each factor $e^{-iv\widetilde\alpha_PP}$ is a Pauli rotation: single-qubit basis changes turn $P$ into a product of $\mathsf Z$ operators, a ladder of at most $q$ controlled-NOT gates collects their parity on one qubit, a single-qubit rotation $e^{-iv\widetilde\alpha_P\mathsf Z}$ acts there, and the ladder and basis changes are undone; the identity string contributes only a phase.
	The circuit is therefore a stream of $nQ^2$ Pauli rotations of $\mathcal O(q)$ gates each, and the controller stores only the step counter, the current string $P$, and $\widetilde\alpha_P$, using $\mathcal O(\chi\log N+\log(t+\|f\|_1+1/\epsilon))$ bits, besides the workspace of \cref{lem:compression-pauli}.
	Finally, each Pauli rotation is compiled into the fixed gate set to error $\epsilon/(3nQ^2)$ in time polylogarithmic and workspace logarithmic in $nQ^2/\epsilon$ \cite{vanmelkebeek2012}; the compilation errors telescope to at most $\epsilon/3$, and the compilation phases combine into $\theta$.
	Rounding, product formula and compilation thus contribute $\epsilon/3$ each, for a total error of at most $\epsilon$.
	With $Q\le(2N+2)^{\chi-1}=N^{\mathcal O(\chi)}$ and $\log(1/\delta)=\mathcal O(\chi\log N+\log(t+1/\epsilon))$, the $nQ^2$ coefficients take time $\mathrm{poly}(t,\|f\|_1,1/\epsilon)\,N^{\mathcal O(\chi)}\chi^{\mathcal O(k)}$ and every workspace is $\mathcal O\big(\chi\log N+k\log\chi+\log(t+\|f\|_1+1/\epsilon)\big)$ bits.
\end{proof}

A piecewise-constant schedule, in which $H$ takes values $H_1,\ldots,H_S$ of the form \cref{eq:compression-normalized-family} on $S=\mathrm{poly}(N)$ consecutive intervals with logarithmic-bit rational endpoints, is implemented to error $\epsilon$ by applying the lemma to each interval with error $\epsilon/S$.

\begin{restatedtheorem}{thm:informal-quantum}[Logarithmic-workspace PI dynamics estimation]
	\label{thm:compression-dynamics}
	Let $H=Nf$ be a $k$-local, permutation-invariant Hamiltonian of local dimension $\chi$ supplied as in \cref{lem:compression-element-access}, with $\|f\|_1\le\mathrm{poly}(N)$, or a piecewise-constant schedule of $S\le\mathrm{poly}(N)$ such Hamiltonians, and let $t\le\mathrm{poly}(N)$ be a logarithmic-bit rational.
	Let $\rho(0)=(\ket\phi\bra\phi)^{\otimes N}$ be a symmetric product state whose single-site amplitudes are given to inverse-polynomial precision, and let $O$ be an observable supported on at most $k_O$ sites with $\|O\|\ge1/\mathrm{poly}(N)$, supplied as in \cref{lem:local-element-access}.
	Then for $0<\epsilon<1$ with $\epsilon^{-1}\le\mathrm{poly}(N)$, a quantum algorithm solves \cref{prob:averaged-dynamics}, i.e.\ it outputs $\widetilde F$ with
	\begin{align}
		\Pr\!\left[|\widetilde F-F_O(t)|\le\epsilon\|O\|\right]\ge\frac23,
	\end{align}
	using, simultaneously, $\mathrm{poly}(t,1/\epsilon)\,N^{\mathcal O(\chi)}\chi^{\mathcal O(k+k_O)}$ time, $\mathcal O(\chi\log N)$ qubits and $\mathcal O(\chi\log N+(k+k_O)\log\chi+\log(t+1/\epsilon))$ bits of classical workspace.
	In particular, for $\chi=\mathcal O(1)$ and $k,k_O=\mathcal O(\log N)$, this is polynomial time and $\mathcal O(\log N)$ qubits and workspace.
\end{restatedtheorem}

\begin{proof}
	The initial state is prepared by evolving the all-zero register for unit time under a one-local permutation-invariant Hamiltonian.
	Up to a global phase, $\ket\phi=\cos\vartheta\ket0+\sin\vartheta\ket v$ with $\vartheta\in[0,\pi/2]$ and a unit vector $\ket v\perp\ket0$, and the single-site generator $B=i\vartheta(\ket v\bra0-\ket0\bra v)$ rotates $\ket0$ into $e^{-iB}\ket0=\ket\phi$.
	As the all-zero register encodes $\ket0^{\otimes N}=V\ket{0^q}$, \cref{eq:compression-intertwining} gives
	\begin{align}
		\ket{\psi_0} & =\big(e^{-iB}\big)^{\otimes N}\ket0^{\otimes N}=e^{-iNm(B)}V\ket{0^q}\nonumber \\
		             & =Ve^{-ib}\ket{0^q}=V\ket{\phi_0},
	\end{align}
	with $b=V^\dagger Nm(B)V$ and $\ket{\phi_0}=e^{-ib}\ket{0^q}$.
	For $\vartheta>0$, $Nm(B)=N\vartheta\,m(B/\vartheta)$ is of the form \cref{eq:compression-normalized-family} with the single monomial $m(B/\vartheta)$, $\mathcal B=\{B/\vartheta\}$ and $\|f\|_1=\vartheta\le\pi/2$, where the entries of $B/\vartheta$ are $\pm i$ times amplitudes of $\ket v$, computed from those of $\ket\phi$ whenever \cref{lem:compression-element-access} reads them.
	Hence \cref{prop:compression-simulation} prepares $\ket{\phi_0}$ as one extra interval of unit length before the schedule.

	The compressed observable $o=V^\dagger OV$ is a dense $Q\times Q$ matrix and cannot be measured directly, so we expand it in Pauli strings and measure a random one.
	By \cref{eq:compression-pauli}, $o=\sum_P\alpha_PP$ with $\sum_P\alpha_P^2\le\|o\|^2\le\|O\|^2$, and by \cref{lem:compression-pauli}, applied to $o$ with its entries computed by \cref{lem:local-element-access}, each $\alpha_P$ can be computed to error $\delta=\epsilon\|O\|/(4Q^2)$; let $\widetilde\alpha_P$ be the result.
	One repetition of the algorithm draws $P\in\mathcal P_q$ uniformly at random, prepares $\ket{\phi_0}$ and evolves it by \cref{prop:compression-simulation}, measures $P$ by single-qubit basis changes followed by a computational-basis measurement, with outcome $X_P=\pm1$ the parity on the support of $P$, and outputs $Y=Q^2\widetilde\alpha_PX_P$.
	With $x_P=\mathbb E[X_P]$ the expected outcome of the implemented circuit,
	\begin{align}
		\mathbb E[Y] & =\sum_P\widetilde\alpha_Px_P,\nonumber\\
		\mathbb E[Y^2] & =Q^2\sum_P\widetilde\alpha_P^2\le Q^2\big(\|O\|+Q\delta\big)^2\le4Q^2\|O\|^2,
		\label{eq:compression-pauli-sampling}
	\end{align}
	by the triangle inequality in $\ell^2$ for $\widetilde\alpha=\alpha+(\widetilde\alpha-\alpha)$.

	For an exact circuit, $x_P=\tr(P\varrho(t))$, and hence $\sum_P\alpha_Px_P=\tr(o\varrho(t))=F_O(t)$ by \cref{sec:compressed-encoding}.
	The error thus has three sources: the rounded coefficients, the imperfect circuit, and the finite number of repetitions.
	Implement the preparation, the evolution (at error $\eta/(3S)$ per interval for a schedule), and the basis change ($q$ single-qubit gates at error $\eta/(3q)$ each) each to error $\eta/3$, with $\eta=\epsilon/(16Q)$.
	The state vector before the computational-basis measurement is then within $\eta$ of the ideal one up to a phase, so $|x_P-\tr(P\varrho(t))|\le2\eta$.
	With $\sum_P|\widetilde\alpha_P|\le Q\big(\sum_P\widetilde\alpha_P^2\big)^{1/2}\le2Q\|O\|$ and $|\tr(P\varrho(t))|\le1$,
	\begin{align}
		\big|\mathbb E[Y]-F_O(t)\big| & \le\sum_P|\widetilde\alpha_P-\alpha_P|\nonumber\\
		                             & \quad+\sum_P|\widetilde\alpha_P|\,\big|x_P-\tr(P\varrho(t))\big|\nonumber\\
		                             & \le Q^2\delta+4Q\|O\|\eta=\frac{\epsilon\|O\|}2.
	\end{align}
	Averaging $n=\lceil48Q^2/\epsilon^2\rceil$ independent repetitions gives $\widetilde F$ with $\Pr\big[|\widetilde F-\mathbb E[Y]|\ge\epsilon\|O\|/2\big]\le4\,\mathbb E[Y^2]/(n\epsilon^2\|O\|^2)\le1/3$ by Chebyshev's inequality and \cref{eq:compression-pauli-sampling}, and hence $|\widetilde F-F_O(t)|\le\epsilon\|O\|$ with probability at least $2/3$.

	Since $\|f\|_1,\|O\|^{-1}\le\mathrm{poly}(N)$, both $1/\eta$ and $1/\delta$ are $N^{\mathcal O(\chi)}/\epsilon$.
	Hence, by \cref{prop:compression-simulation,lem:compression-pauli,lem:local-element-access}, the circuits for the preparation and for each of the $S\le\mathrm{poly}(N)$ intervals, as well as the coefficient $\widetilde\alpha_P$, are generated in time $\mathrm{poly}(t,1/\epsilon)\,N^{\mathcal O(\chi)}\chi^{\mathcal O(k+k_O)}$ and workspace $\mathcal O(\chi\log N+(k+k_O)\log\chi+\log(t+1/\epsilon))$, where the preparation computes the entries of $B/\vartheta$ from the amplitudes of $\ket\phi$ in logarithmic workspace.
	The $n=N^{\mathcal O(\chi)}/\epsilon^2$ repetitions thus take time $\mathrm{poly}(t,1/\epsilon)\,N^{\mathcal O(\chi)}\chi^{\mathcal O(k+k_O)}$.
	They run one after another on the same reset register of $q=\mathcal O(\chi\log N)$ qubits and in the same classical workspace, where the controller additionally keeps the repetition counter, $P$, and the running sum of the outputs $Y$, all of $\mathcal O(\chi\log N+\log(1/\epsilon))$ bits.
\end{proof}

\subsection{Hardness of Permutation-Invariant Hamiltonian Simulation}
\label{sec:hardness}
We now show the complementary direction: simulating permutation-invariant dynamics is at least as hard as running a quantum circuit on logarithmically many qubits. Concretely, we use the qubit encoding of \cref{sec:compressed-encoding} to embed a circuit on $q_{\mathrm c}$ qubits in the symmetric subspace of $N\ge2^{q_{\mathrm c}}-1$ spins, and compile it, in polynomial time and logarithmic classical workspace, into a piecewise-constant schedule of permutation-invariant two-local Hamiltonians of the form \cref{eq:compression-normalized-family} whose evolution reproduces the circuit on the encoded states.

We use the notation of \cref{sec:compressed-encoding} with $\chi=2$, so that $V$ maps the binary label $0\le n\le N$ of the $q=\lceil\log_2(N+1)\rceil$ register qubits to the Dicke state $\ket n_N$ of \cref{eq:compression-dicke}.
A circuit on $q_{\mathrm c}$ qubits acts on the $q_{\mathrm c}$ least significant register qubits, the remaining ones being fixed to zero, and hence on the labels $a<d=2^{q_{\mathrm c}}$; choosing $N\ge d-1$, in particular $N=d-1$, ensures that all these labels are valid occupations.
Without loss of generality we assume the circuit is given over the universal (see \cref{app:gate-universality}) gate set
\begin{equation}
	\left\{A_i=e^{-i\pi\mathsf X_i/4},
	B_i=e^{-i\pi\mathsf Y_i/8},
	C_{ij} =e^{-i\pi\mathsf X_i\mathsf X_j/4}	\right\}
	\label{eq:hardness-gates}
\end{equation}
and their inverses.

\begin{lemma}[Circuit-to-control compilation]
	\label{lem:circuit-to-control}
	Given a circuit $U$ of $L\ge1$ gates over
	\cref{eq:hardness-gates} on $q_{\mathrm c}$ qubits, an integer $N\ge2^{q_{\mathrm c}}-1$, and an error
	tolerance $0<\eta\le1$, a classical compiler outputs
	bounded piecewise-constant controls of the form of 2-local permutation-invariant qubit Hamiltonian
  \begin{equation}
    \begin{gathered}
      H(t)=N\big(h_x(t)m_x+h_y(t)m_y+h_z(t)m_z
      +h_2(t)m_z^2\big),\\
      |h_\alpha(t)|\le1,\qquad \alpha\in\{x,y,z,2\},
    \end{gathered}
    \label{eq:hardness-controls}
  \end{equation}
  with $h_2=1/4$ throughout,
	whose time-ordered evolution $U_{\mathrm{phys}}$ satisfies
	\begin{equation}
		\|U_{\mathrm{phys}}V-VU\|\le\eta.
		\label{eq:hardness-intertwining}
	\end{equation}
	The compiler runs in time $\mathcal O(LN^2\,\mathrm{polylog}(NL/\eta))$ with
	$\mathcal O(\log(NL/\eta))$ classical workspace.
	The schedule has at most $n^{\mathrm{PR}}=Ld^2$ segments,
	total duration $\mathcal O(N^2(n^{\mathrm{PR}})^2/\eta)$, and
	logarithmic-bit parameters.
\end{lemma}

\paragraph{Selective pulses}
We first implement a rotation on one chosen pair of Dicke states with neighbouring occupations, leaving all other Dicke states unchanged. For $0\le a<b\le N$, let $S^x_{a,b}$ and $S^y_{a,b}$ act on the ordered basis $(\ket{a}_N,\ket{b}_N)$ as $\begin{psmallmatrix}0&1\\1&0\end{psmallmatrix}$ and $\begin{psmallmatrix}0&-i\\i&0\end{psmallmatrix}$, and as zero on all other Dicke states, and let $R_\alpha^{a,b}(\theta)=e^{-i\theta S^\alpha_{a,b}}$ for $\alpha\in\{x,y\}$, that is,
\begin{align}
  R_x^{a,b}(\theta)\big|_{(a,b)}
  &=\begin{pmatrix}\cos\theta&-i\sin\theta\\-i\sin\theta&\cos\theta\end{pmatrix},\nonumber\\
  R_y^{a,b}(\theta)\big|_{(a,b)}
  &=\begin{pmatrix}\cos\theta&-\sin\theta\\\sin\theta&\cos\theta\end{pmatrix},
  \label{eq:hardness-pair-operators}
\end{align}
and identity on all other Dicke states.

\begin{lemma}[Selective pulse]
  \label{lem:selective-pulse}
  Let $N,R\ge1$ be integers, $\ell\in\{0,\ldots,N-1\}$, $\alpha\in\{x,y\}$, and $\theta\in[-\pi,\pi]$. The constant Hamiltonian
  \begin{align}
    H_{\ell,\alpha,\theta}
    &=N\bigg(\frac{m_z^2}{4}+\frac{2\ell+1-N}{2N}m_z\nonumber\\
    &\quad+\frac{\theta}{8\pi NR\sqrt{(\ell+1)(N-\ell)}}m_\alpha\bigg),
    \label{eq:hardness-pulse}
  \end{align}
  applied for time $\tau=8\pi NR$, has control amplitudes bounded by one and satisfies
  \begin{equation}
    \left\|e^{-i\tau H_{\ell,\alpha,\theta}}-R_\alpha^{\ell,\ell+1}(\theta)\right\|
    \le\frac{\theta^2N}{2\pi R}.
    \label{eq:hardness-pulse-error}
  \end{equation}
  The norm is the operator norm restricted to the symmetric subspace.
\end{lemma}

The proof is given in \cref{sec:proofs}; its mechanism is resonance.
The diagonal part $Q_\ell=\frac N4m_z^2+\frac{2\ell+1-N}2m_z$ of $H_{\ell,\alpha,\theta}$ has eigenvalues $\mu_n$ on $\ket n_N$ with $\mu_{n+1}-\mu_n=2(n-\ell)/N$ and $\tau\mu_n\in2\pi\mathbb Z$, so the selected transition $\ell\leftrightarrow\ell+1$ is resonant, while every other transition is detuned by at least $2/N$ and completes an integer number of full cycles over the pulse.

\paragraph{Pairing and routing}
A logical gate mixes $d/2$ disjoint pairs of Dicke states, whereas each pair rotation acts on a single pair. Each gate in \cref{eq:hardness-gates} flips the bits in a mask $M$ and rotates about an axis $\alpha$ by an angle $\theta$,
\begin{equation}
  (M,\alpha,\theta)=
  \begin{cases}
    (2^i,\;x,\;\pi/4) & \text{for }A_i,\\
    (2^i,\;y,\;\pi/8) & \text{for }B_i,\\
    (2^i+2^j,\;x,\;\pi/4) & \text{for }C_{ij},
  \end{cases}
  \label{eq:hardness-pairing-rule}
\end{equation}
with $\theta$ negated for the inverse gate. The mask pairs each label $a$ with $a\oplus M$, where $\oplus$ denotes bitwise XOR.

\begin{lemma}[Pairing]
  \label{lem:pairing}
  For a gate $G$ with parameters $(M,\alpha,\theta)$ as in \cref{eq:hardness-pairing-rule},
  \begin{equation}
    VG=\Bigg(\prod_{\substack{0\le a<d\\a<a\oplus M}}
       R_\alpha^{a,a\oplus M}(\theta)\Bigg)V.
    \label{eq:hardness-pairing}
  \end{equation}
  The $d/2$ factors act on disjoint pairs and commute.
\end{lemma}
\begin{proof}
  The generators satisfy $\mathsf X_i\ket{a}=\ket{a\oplus2^i}$, $\mathsf X_i\mathsf X_j\ket{a}=\ket{a\oplus(2^i+2^j)}$, and $\mathsf Y_i\ket{a}=i(-1)^{a_i}\ket{a\oplus2^i}$, where $a_i$ is bit $i$ of $a$. Each generator is therefore block diagonal with $2\times2$ blocks on the pairs $\{a,a\oplus M\}$. On the ordered pair $(\ket a,\ket{a\oplus M})$ with $a<a\oplus M$, the block is $\begin{psmallmatrix}0&1\\1&0\end{psmallmatrix}$ for $A_i$ and $C_{ij}$, and, since $a<a\oplus2^i$ means $a_i=0$, $\begin{psmallmatrix}0&-i\\i&0\end{psmallmatrix}$ for $B_i$; these are the blocks of $S^\alpha_{a,a\oplus M}$ on $(\ket a_N,\ket{a\oplus M}_N)$. Exponentiating block by block, $G$ acts on each ordered pair as $R_\alpha^{a,a\oplus M}(\theta)=e^{-i\theta S^\alpha_{a,a\oplus M}}$ acts on $(\ket a_N,\ket{a\oplus M}_N)$. Since the pairs are disjoint, the blocks commute and the factorization introduces no phase.
\end{proof}

The paired occupations need not be adjacent. To rotate a pair with nonadjacent occupations $a<b$, move $\ket b_N$ to $\ket{a+1}_N$ without disturbing $\ket a_N$, apply the adjacent rotation, and undo the moves. Each move is a signed swap
\begin{gather}
  W_j=R_y^{j,j+1}(-\pi/2),\\
  W_j\ket{j}_N=-\ket{j+1}_N,\qquad
  W_j\ket{j+1}_N=\ket{j}_N.
\end{gather}

\begin{lemma}[Routing]
  \label{lem:routing}
  For $0\le a<b<d$, let $P_{a,b}=W_{a+1}W_{a+2}\cdots W_{b-1}$. Then
  \begin{equation}
    R_\alpha^{a,b}(\theta)
    =P_{a,b}^\dagger R_\alpha^{a,a+1}(\theta)P_{a,b},
    \label{eq:hardness-routing}
  \end{equation}
  a product of $2(b-a-1)+1<2d$ adjacent pair rotations.
\end{lemma}
\begin{proof}
  Operators act right to left, so $P_{a,b}$ applies $W_{b-1},W_{b-2},\ldots,W_{a+1}$ in turn. It maps $\ket{b}_N$ to $\ket{a+1}_N$ with positive sign and fixes $\ket{a}_N$, which none of the factors touch. Hence the right-hand side of \cref{eq:hardness-routing} acts on $(\ket{a}_N,\ket{b}_N)$ as $R_\alpha^{a,a+1}(\theta)$ acts on $(\ket{a}_N,\ket{a+1}_N)$. Every other Dicke state is mapped by $P_{a,b}$ to a signed Dicke state with occupation outside $\{a,a+1\}$, left unchanged by the central rotation, and restored by $P_{a,b}^\dagger$.
\end{proof}

\Cref{lem:pairing} for each gate and \cref{lem:routing} for each of its $d/2$ pairs give $VU=\widetilde UV$, where $\widetilde U$ is a product of at most $n^{\mathrm{PR}}=Ld^2$ adjacent pair rotations, and the compiler realises each of them by the pulse of \cref{lem:selective-pulse} with
\begin{equation}
	R=\left\lceil\frac{4Nn^{\mathrm{PR}}}{\eta}\right\rceil.
	\label{eq:hardness-resource-choice}
\end{equation}
The error accounting and the resource count that complete the proof of \cref{lem:circuit-to-control} are given in \cref{sec:proofs}.

We now consider the consequences of the above for classical simulation.

\begin{restatedtheorem}{cor:informal-classical-hardness}[Conditional classical hardness]
	\label{thm:classical-hardness}
	Fix a constant $0<\epsilon<1/3$, and let $\mathcal A$ be a classical algorithm that solves \cref{prob:averaged-dynamics} to accuracy $\epsilon$ (for $J=0$, $O=\sigma^z_1$ and $\rho(0)=\ket0\bra0^{\otimes N}$) with $H_{J=0}$ a piecewise-constant schedule of (two-local) permutation-invariant Hamiltonians and $t$ its total duration.
 	\begin{enumerate}
		\item[(a)] If $\mathcal A$ is deterministic and runs in time $\mathrm{poly}(N,t)$ and, simultaneously, space $\mathrm{polylog}(N,t)$, then $\Lclass\subseteq\SC$.
		\item[(b)] If $\mathcal A$ is randomised, succeeds with probability at least $2/3$, and runs in time $\mathrm{poly}(N,t)$ and, simultaneously, space $\mathcal O(\log(Nt))$, then $\Lclass=\BPL$.
	\end{enumerate}
\end{restatedtheorem}

\begin{proof}
	Let $\mathcal L\in\Lclass$, let $x$ be an input of length $n$, and let $\gamma=1/3-\epsilon>0$.
	Since $\ket0^{\otimes N}$ and every Hamiltonian of the schedule are permutation invariant, the evolved state lies in the symmetric subspace, where $\sigma^z_1$ and $m_z$ have the same compressed matrix (\cref{sec:element-access-local-observable}); hence $\mathcal A$ estimates $F_{\sigma^z_1}(t)=F_{m_z}(t)$, with $\|\sigma^z_1\|=\|m_z\|=1$.
	From $x$ we construct a schedule whose magnetisation $F_{m_z}(t)$ lies within $\gamma/2$ of the polarisation $1-2p_x$ of the output qubit of a quantum circuit deciding $\mathcal L$; the sign of the estimate of $\mathcal A$ then decides whether $x\in\mathcal L$.

	By \citet{fefferman2021}, intermediate measurements can be removed from $\Lclass$ computations within the same resources.
	Hence a deterministic machine running in polynomial time and $\mathcal O(\log n)$ space emits, on input $x$, the $\mathrm{poly}(n)$ gates of a circuit $U_x$ over a fixed finite gate set on $\mathcal O(\log n)$ qubits.
	Measuring the output qubit of $U_x$ applied to the all-zero state gives one with probability $p_x$, where $p_x\ge2/3$ if $x\in\mathcal L$ and $p_x\le1/3$ otherwise, so $1-2p_x\le-1/3$ if $x\in\mathcal L$ and $1-2p_x\ge1/3$ otherwise.

	For a register of $q_{\mathrm c}$ qubits we take $N=2^{q_{\mathrm c}}-1$, so that $V$ maps the register unitarily onto the symmetric subspace, the label $a$ going to the Dicke state $\ket a_N$ with $a$ spins in state $\ket1$.
	Thus $m_z=1-2a/N$ on $\ket a_N$; writing $a=\sum_i2^ia_i$ in binary and using $\sum_{i<q_{\mathrm c}}2^i=N$, this is $\sum_i(2^i/N)(1-2a_i)$, that is,
	\begin{equation}
		V^\dagger m_zV=\sum_{i=0}^{q_{\mathrm c}-1}\frac{2^i}{N}\,\mathsf Z_i.
	\end{equation}
	The most significant qubit carries only about half of the weight, so we copy the output qubit.

	Let $c=\lceil\log_2(8/\gamma)\rceil$, a constant, and let $U$ be $U_x$ followed by $c-1$ controlled-NOT gates that copy its output qubit onto fresh qubits.
	We order the $q_{\mathrm c}=\mathcal O(\log n)$ qubits of $U$ such that the $c$ qubits carrying the output bit are the most significant, so that their total weight is $w\ge1-2^{-c}$.
	In $\ket\varphi=VU\ket{0^{q_{\mathrm c}}}$ each of them has $\langle\mathsf Z\rangle=1-2p_x$ and every other qubit has $|\langle\mathsf Z\rangle|\le1$, hence
	\begin{equation}
		\big|\bra\varphi m_z\ket\varphi-(1-2p_x)\big|\le2(1-w)\le2^{1-c}\le\gamma/4.
	\end{equation}

	It remains to realise $U$ as a schedule.
	Compiling each gate of $U$ into the gate set \cref{eq:hardness-gates} to error $\gamma/16$ divided by the number of gates, in polynomial time and logarithmic space~\cite{vanmelkebeek2012}, gives a circuit $U'$ of $L=\mathrm{poly}(n)$ gates with $\|U'-U\|\le\gamma/16$.
	\Cref{lem:circuit-to-control} applied to $U'$, $N$, and $\eta=\gamma/16$ yields a schedule of the form \cref{eq:hardness-controls} with $S,t=\mathrm{poly}(n)$ whose evolution $U_{\mathrm{phys}}$ satisfies $\|U_{\mathrm{phys}}V-VU'\|\le\gamma/16$.
	The initial state is $\ket0^{\otimes N}=\ket0_N=V\ket{0^{q_{\mathrm c}}}$, so the evolved state is within $\gamma/8$ of $\ket\varphi$; as expectation values in unit vectors differ by at most $2\|m_z\|=2$ times the vector distance, $|F_{m_z}(t)-\bra\varphi m_z\ket\varphi|\le\gamma/4$.

	Altogether, $|\widetilde F-(1-2p_x)|\le\epsilon+\gamma/2<1/3$, so $\widetilde F<0$ if and only if $x\in\mathcal L$.

	The instance $(N,\text{schedule})$ is produced from $x$ by composing the machines above, each running in polynomial time and $\mathcal O(\log n)$ space; we run $\mathcal A$ on it and recompute each input symbol whenever $\mathcal A$ reads it, which keeps both bounds.
	As $N,t,S=\mathrm{poly}(n)$, in case (a) the procedure is deterministic and runs in polynomial time and polylogarithmic space, so $\mathcal L\in\SC$; in case (b) it errs with probability at most $1/3$ and runs in polynomial time and $\mathcal O(\log n)$ space, so $\mathcal L\in\BPL$, and $\Lclass=\BPL$ follows as $\BPL\subseteq\Lclass$ trivially.
\end{proof}

\subsection{Classical Sampling for \texorpdfstring{$\mathcal{O}(\log N/N)$}{O(log N/N)} Times}
\label{sec:classical-sampling}

We now turn from estimating expectation values to sampling: for a permutation-invariant Hamiltonian $H=Nf$ on $N$ sites of local dimension $\chi$, supplied as in \cref{lem:compression-element-access}, and a symmetric initial state $\rho(0)$, draw a string $x\in\{0,\ldots,\chi-1\}^N$ from the computational-basis measurement distribution of the evolved state,
\begin{equation}
  P(x)=\bra{x}e^{-iHt}\rho(0)e^{iHt}\ket{x}.
  \label{eq:sampling-task}
\end{equation}
We show that for times $t=\mathcal O(\log N/N)$ the task is solved by a classical algorithm in $\mathcal O(\chi\log N)$ space.

Since $\rho(0)$ is supported on the symmetric subspace, \cref{eq:compression-intertwining} gives $\rho(t)=V\varrho(t)V^\dagger$ with $\varrho(t)=e^{-iht}\varrho(0)e^{iht}$, and as $\braket{x|\mathbf n}_N=M_N(\mathbf n)^{-1/2}$ for $x\in\mathcal X_{\mathbf n}$ and zero otherwise, a string $x\in\mathcal X_{\mathbf n}$ has
$
  P(x)=\varrho(t)_{\mathbf n\mathbf n} / M_N(\mathbf n),
$
cf. \cref{eq:generalised-dicke-states}.
The measurement thus factorises into two classical stages: draw an occupation $\mathbf n$ with probability $\varrho(t)_{\mathbf n\mathbf n}$, then a uniformly random string in $\mathcal X_{\mathbf n}$.
The accuracy of a sampler with output distribution $\widetilde P$ is measured by the total variation distance $d_{\mathrm{TV}}(P,\widetilde P)=\frac12\sum_x|P(x)-\widetilde P(x)|$.

\begin{restatedtheorem}{thm:informal-sampling}[Classical sampling of permutation-invariant dynamics]
\label{thm:classical-sampling}
Let $H=Nf$ be a $k$-local, permutation-invariant Hamiltonian of local dimension $\chi$ supplied as in \cref{lem:compression-element-access}, with $k$ and $\|f\|_1$ independent of $N$, let $t=\mathcal O(\log N/N)$ and $\epsilon^{-1}\le\mathrm{poly}(N)$, and let $\rho(0)$ be supported on the symmetric subspace with a compressed matrix $\varrho(0)=V^\dagger\rho(0)V$ whose entries are computable to inverse-polynomial precision in polynomial time and logarithmic space.
Then a classical algorithm using random bits outputs, symbol by symbol, a string whose distribution is within total variation distance $\epsilon$ of $P$, in time $N^{\mathcal O(\chi)}$ and, simultaneously, $\mathcal O(\chi\log N)$ workspace.
\end{restatedtheorem}

\begin{proof}
By \cref{eq:compression-matrix-units}, every collective operator changes each occupation $n_a$ by at most one, so each monomial of $f$ changes it by at most $k$, and $h$ is banded: $h_{\mathbf m\mathbf n}=0$ for $\|\mathbf m-\mathbf n\|_\infty>k$.
Truncate the exponential series,
\begin{equation}
  U_L=\sum_{\ell=0}^L\frac{(-it)^\ell}{\ell!}h^\ell,\qquad
  \|e^{-iht}-U_L\|\le\frac{\|th\|^{L+1}}{(L+1)!}\eqqcolon\delta.
  \label{eq:sampling-truncation}
\end{equation}
Since $\|th\|\le tN\|f\|_1=\mathcal O(\log N)$ by \cref{lem:compression-element-access} and $\epsilon^{-1}\le\mathrm{poly}(N)$, an order $L=\mathcal O(\log N)$ achieves $\delta\le\epsilon/8$.
An entry $(h^\ell)_{\mathbf m\mathbf n}$ is the sum over paths of occupations $\mathbf n=\mathbf n_0,\mathbf n_1,\ldots,\mathbf n_\ell=\mathbf m$ with $\|\mathbf n_j-\mathbf n_{j-1}\|_\infty\le k$ of the products $h_{\mathbf n_\ell\mathbf n_{\ell-1}}\cdots h_{\mathbf n_1\mathbf n_0}$.
A path is fixed by its steps in the $\chi-1$ stored occupations $n_1,\ldots,n_{\chi-1}$, and paths leaving the valid occupations contribute zero.
For $\ell\le L$ there are thus at most $(2k+1)^{(\chi-1)L}=N^{\mathcal O(\chi)}$ such paths, each specified by $\mathcal O(\chi\log N)$ bits, so $(U_L)_{\mathbf m\mathbf n}$ is computed by enumerating them and accumulating, with the factor $-it/j$ absorbed at the $j$th multiplication.
The entries of $h$ are computed on the fly in time $N^{\mathcal O(\chi)}$ and workspace $\mathcal O(\chi\log N)$ by \cref{lem:compression-element-access}; every partial product is bounded by $\|th\|^j/j!\le e^{\|th\|}=\mathrm{poly}(N)$, and rounding errors are amplified by at most $N^{\mathcal O(\chi)}$ factors, so arithmetic on $\mathcal O(\chi\log N)$ bits yields accuracy $N^{-c\chi}$ for any constant $c$.
The same applies, with the entries of $\varrho(0)$ computed to precision $N^{-c\chi}$ by assumption, to the approximate weights
\begin{equation}
  p_{\mathbf n}=\big(U_L\varrho(0)U_L^\dagger\big)_{\mathbf n\mathbf n}
  =\sum_{\mathbf a,\mathbf b}(U_L)_{\mathbf n\mathbf a}\,\varrho(0)_{\mathbf a\mathbf b}\,\overline{(U_L)_{\mathbf n\mathbf b}},
  \label{eq:sampling-weights}
\end{equation}
a sum of at most $Q^2$ terms, with $Q\le(2N+2)^{\chi-1}$ the register dimension of \cref{sec:compressed-encoding}, that are recomputed rather than stored.
Writing $U_L=e^{-iht}+D$ with $\|D\|\le\delta$, the sum of the absolute diagonal entries of a matrix is at most its trace norm, and $\|AXB\|_1\le\|A\|\|X\|_1\|B\|$ with $\|\varrho\|_1=1$ gives
\begin{equation}
  \sum_{\mathbf n}|p_{\mathbf n}-\varrho(t)_{\mathbf n\mathbf n}|\le\|U_L\varrho(0)U_L^\dagger-\varrho(t)\|_1\le2\delta+\delta^2.
  \label{eq:sampling-weight-error}
\end{equation}
Let $w_{\mathbf n}\ge0$ be $p_{\mathbf n}$ computed to accuracy $\epsilon/(8Q)$ and clipped at zero; as $\varrho(t)_{\mathbf n\mathbf n}\ge0$, clipping does not increase the error, and $\sum_{\mathbf n}|w_{\mathbf n}-\varrho(t)_{\mathbf n\mathbf n}|\le\epsilon/2$.
Since $\sum_{\mathbf n}\varrho(t)_{\mathbf n\mathbf n}=1$, normalising the weights changes them by at most another $\epsilon/2$ in total, so the distribution $w_{\mathbf n}/\sum_{\mathbf m}w_{\mathbf m}$ is within total variation distance $\epsilon/2$ of $\varrho(t)_{\mathbf n\mathbf n}$.

For the first stage, the weights are never stored: scan the valid occupations $\mathbf n$ in lexicographic order of $(n_1,\ldots,n_{\chi-1})$, keeping a running total $S$ and a candidate $\mathbf n^\ast$: on reaching $w_{\mathbf n}>0$, set $\mathbf n^\ast=\mathbf n$ with probability $w_{\mathbf n}/(S+w_{\mathbf n})$, and then replace $S$ by $S+w_{\mathbf n}$.
If $\mathbf n^\ast=\mathbf m$ held with probability $w_{\mathbf m}/S$ before this step, it holds with probability $\frac{w_{\mathbf m}}{S}\big(1-\frac{w_{\mathbf n}}{S+w_{\mathbf n}}\big)=\frac{w_{\mathbf m}}{S+w_{\mathbf n}}$ afterwards, so by induction the final candidate equals $\mathbf n$ with probability $w_{\mathbf n}/\sum_{\mathbf m}w_{\mathbf m}$.
Beyond the workspace for one weight, only $\mathbf n$, $\mathbf n^\ast$, $w_{\mathbf n}$, and $S$ are kept, $\mathcal O(\chi\log N)$ bits each.

For the second stage, given $\mathbf n^\ast=\mathbf n$, output the symbols $x_1,\ldots,x_N$ in turn, choosing $x_i=a$ with probability equal to the number of sites still to be put in state $a$ divided by the number of positions left.
For each $a$ the numerators run through $n_a,n_a-1,\ldots,1$ and the denominators through $N,N-1,\ldots,1$, so every string in $\mathcal X_{\mathbf n}$ is output with probability $\prod_an_a!/N!=1/M_N(\mathbf n)$.
This produces a uniformly random string in $\mathcal X_{\mathbf n}$ and stores only the $\chi$ remaining counts.

Both stages together make at most $Q+N$ random choices.
Realising each with $\mathcal O(\log(Q/\epsilon))=\mathcal O(\chi\log N)$ random bits to accuracy $\epsilon/(2(Q+N))$ perturbs the output distribution by at most $\epsilon/2$ in total variation, giving total error $\epsilon$ in time $N^{\mathcal O(\chi)}$ and $\mathcal O(\chi\log N)$ workspace.
\end{proof}

\subsection{Classical Dynamics for \texorpdfstring{$\mathcal{O}(1)$}{O(1)} Times via Cluster Expansion}
\label{sec:cluster-method}

We now leave the symmetric subspace and consider $k$-body all-to-all Hamiltonians on $N$ sites of local dimension $\chi\ge2$ whose terms on up to $k$ sites are arbitrary bounded operators, not necessarily permutation invariant, either fixed or with independent Gaussian couplings.
We show that at constant times $t=\mathcal O(1)$, expectation values of local observables in product states can be estimated classically to constant accuracy in polynomial time and logarithmic space.

\begin{restatedtheorem}{thm:informal-classical}[Classical simulation of all-to-all dynamics at constant times]\label{thm:classical}
Let $k\ge2$ and $k_O\ge1$, and set $\kappa\coloneqq\max\{k,k_O\}$.
Consider $N$ sites of local dimension $\chi$, a product state $\rho(0)=\bigotimes_{i=1}^N\rho_i$, an observable $O$ supported on at most $k_O$ sites, and the Hamiltonian
\begin{equation}\label{eq:classical-hamiltonian}
  H=\sum_{1\le|T|\le k}\frac{h_T}{N^{|T|-1}}+\frac1{N^{(k-1)/2}}\sum_{|T|=k}J_Th'_T,
\end{equation}
where $T$ ranges over subsets of $\{1,\ldots,N\}$, $J_T\overset{\mathrm{i.i.d.}}{\sim}\mathcal N(0,1)$, and $h_T$, $h'_T$ are Hermitian operators on $T$ of operator norm at most $\mathcal J$; the $h'_T$ may vanish, in which case $H$ is deterministic.
Then
\[
  \overline F_O(t)=\mathbb E_J\big[\Tr\big(Oe^{-iHt}\rho(0)e^{iHt}\big)\big]
\]
can be computed to accuracy $\epsilon\|O\|$, $\epsilon\in(0,1)$, by a deterministic classical algorithm in time $(\chi N)^{\mathcal O(kM+k_O)}$ and, simultaneously, in workspace $\mathcal O((kM+k_O)\log(\chi N))$ bits, where $M=\exp\big(\mathcal O(1+\kappa^2(k+1)^{2k}\mathcal J^2t^2)\big)\log(2/\epsilon)$, provided $M\le\mathrm{poly}(N)$ and $t$ and the matrix elements of the $\rho_i$, $h_T$, $h'_T$ and $O$ are available as fixed-point numbers to any required precision.
In particular, for fixed $\chi$, $k$ and $k_O$, $\mathcal J=\mathcal O(1)$, $\epsilon=\Theta(1)$ and $t=\mathcal O(1)$, $M=\mathcal O(1)$ and the computation requires $\mathrm{poly}(N)$ time and $\mathcal O(\log N)$ space.
\end{restatedtheorem}

For $k,k_O=\mathcal O(1)$, \cref{thm:classical} contains \cref{prob:averaged-dynamics} both for $J=0$, as stated in \cref{sec:summary-nondisordered}, and for $r\to\infty$, which proves \cref{thm:sk-classical}.
First, $H_{J=0}=Nf(m_x,m_y,m_z)$ of \cref{eq:hamiltonian} is of the form \cref{eq:classical-hamiltonian} with $h'_T=0$ up to an additive constant: a monomial of degree $\ell\le k$ gives $Nm_{\alpha_1}\cdots m_{\alpha_\ell}=N^{1-\ell}\sum_{i_1,\ldots,i_\ell}\sigma^{\alpha_1}_{i_1}\cdots\sigma^{\alpha_\ell}_{i_\ell}$, and the at most $\ell^\ell\le k^k$ index tuples with a given index set $T$ contribute an operator on $T$ of norm at most $N^{1-|T|}$ each; this gives $\|h_T\|\le k^k\|f\|_1$, with $\|f\|_1$ the coefficient norm of \cref{eq:compression-normalized-family}.
In the limit $r\to\infty$, the disordered part of \cref{eq:hamiltonian} adds $h'_T=-J\sigma^z_{i_1}\cdots\sigma^z_{i_k}$ for $T=\{i_1,\ldots,i_k\}$ with independent Gaussian $J_T$ (\cref{sec:setup}); this includes the $k$-spin models of \cref{eq:models-p-spin} at any transverse field $B$, in particular the SK model \cref{eq:models-sk} for $k=2$.
Hence $\mathcal J=\max\{k^k\|f\|_1,|J|\}$, which is $\mathcal O(1)$ for $\|f\|_1,|J|=\mathcal O(1)$.
Second, \cref{prob:averaged-dynamics} has the product initial state $(\ket\phi\bra\phi)^{\otimes N}$ and an observable supported on at most $k_O$ sites.

The proof of \cref{thm:classical} follows ideas of \citet{wild2023}, and has three ingredients.
First, a cluster expansion shows that the expectation at complex time $s$ is bounded on a region around the real axis whose scale does not shrink with $N$.
Second, analytic continuation from this region expresses the expectation at time $t$ as a truncated Maclaurin series of a composed function.
Third, the truncated series is a finite sum of explicit terms, which can be enumerated by brute force in time $(\chi N)^{\mathcal O(kM+k_O)}$ and space $\mathcal O((kM+k_O)\log(\chi N))$ (\cref{lem:eval}).
Compared to \citet{wild2023}, three changes are needed.
Their analysis assumes an interaction graph whose degree is independent of the system size, so that the convergence radius $1/(2e\Delta)$ of their expansion is a constant; here the degree is $\Delta=\Theta(N^{k-1})$.
For Gaussian couplings, the disorder average is moreover taken inside the expansion by Gaussian integration by parts (\cref{lem:analyticity}), where coinciding labels are paired and the $N^{-(k-1)/2}$ normalisation of the couplings compensates the resulting sums.
Finally, \citet{wild2023} bound only the running time, whereas here the truncated series has to be evaluated in logarithmic space (\cref{lem:eval}).

\subsubsection{Cluster Expansion}
\label{sec:clusterexp}

\paragraph{Reduction to $k$-body terms}
The terms on fewer than $k$ sites can be absorbed into the deterministic $k$-body terms.
Since every set $T$ with $|T|=j\le k$ lies in $\binom{N-j}{k-j}$ sets $S$ with $|S|=k$,
\begin{equation}\label{eq:absorb}
  \begin{gathered}
    \sum_{1\le|T|\le k}\frac{h_T}{N^{|T|-1}}=\frac1{N^{k-1}}\sum_{|S|=k}\tilde h_S,\\
    \tilde h_S=\sum_{\emptyset\ne T\subseteq S}\frac{N^{k-|T|}}{\binom{N-|T|}{k-|T|}}\,h_T
  \end{gathered}
\end{equation}
for $N\ge k$, where $\tilde h_S$ is again Hermitian and supported on $S$, and its matrix elements are computed from those of the $h_T$ in logarithmic space.
Since $\binom{N-j}{k-j}\ge\big(\frac{N-j}{k-j}\big)^{k-j}\ge(N/k)^{k-j}$, $\|\tilde h_S\|\le\sum_{j=1}^k\binom kjk^{k-j}\mathcal J\le(k+1)^k\mathcal J$; for $k=2$, \cref{eq:absorb} reads $\tilde h_{ij}=h_{ij}+\frac N{N-1}(h_{\{i\}}+h_{\{j\}})$.
Nothing below depends on the $k$-body terms beyond their supports and a norm bound.
Replacing $H$ by $H/((k+1)^k\mathcal J)$ and $t$ by $(k+1)^k\mathcal Jt$, we thus assume from now on that $h_T=0$ for $|T|<k$ and $\|h_T\|,\|h'_T\|\le1$.

\paragraph{Expansion in nested commutators}

For a fixed realization, write
\begin{equation}\label{eq:HJ}
  H=\sum_{|S|=k}\left(\frac{h^0_S}{N^{k-1}}+\frac{J_S}{N^{(k-1)/2}}h^1_S\right),
\end{equation}
where $S$ ranges over the $k$-element subsets of $\{1,\ldots,N\}$, and the \emph{type} $\sigma\in\{0,1\}$ of $h^\sigma_S$ distinguishes the deterministic term $h^0_S=h_S$ from the Gaussian term $h^1_S=h'_S$.

For an operator $X$, let $\operatorname{ad}_X$ denote the linear map $Y\mapsto[X,Y]$ on operators, so that $\operatorname{ad}_X^m(O)=[X,[X,\ldots[X,O]]]$ is the $m$-fold nested commutator, and abbreviate $\operatorname{ad}^\sigma_S\coloneqq\operatorname{ad}_{h^\sigma_S}$.
The Heisenberg-evolved observable is the exponential of this map applied to $O$,
\begin{equation}\label{eq:hadamard}
  \begin{aligned}
    e^{iHs}Oe^{-iHs}
    &=e^{is\operatorname{ad}_H}(O)
     =\sum_{m\ge0}\frac{(is)^m}{m!}\operatorname{ad}_H^m(O).
  \end{aligned}
\end{equation}
Since $\operatorname{ad}_H=\sum_{|S|=k}(N^{-(k-1)}\operatorname{ad}^0_S+N^{-(k-1)/2}J_S\operatorname{ad}^1_S)$, its $m$-th power expands over $m$-tuples of sets $S_i$ with types $\sigma_i$,
\begin{equation}\label{eq:ad-expansion}
  \operatorname{ad}_H^m(O)
  =\sum_{\sigma\in\{0,1\}^m}\sum_{S_1,\ldots,S_m}
  \frac{\prod_{i:\sigma_i=1}J_{S_i}}{N^{(k-1)(d+g/2)}}
  \operatorname{ad}^{\sigma_1}_{S_1}\cdots\operatorname{ad}^{\sigma_m}_{S_m}(O),
\end{equation}
where $d$ and $g=m-d$ are the numbers of deterministic and Gaussian types in $\sigma$.
For $s\in\mathbb C$, let
\begin{equation}\label{eq:Os}
  f(s)=\overline F_O(s)=\mathbb E_J\big[\Tr\big(\rho(0)e^{iHs}Oe^{-iHs}\big)\big]
\end{equation}
be the expectation value of \cref{thm:classical} at complex time $s$.
Taking the trace against $\rho(0)$ and the expectation over $J$ term by term, $f$ has the Maclaurin expansion
\begin{equation}\label{eq:taylor}
  \begin{aligned}
    f(s)&=\sum_{m\ge0}a_ms^m,\quad
    w_\sigma=N^{-(k-1)(d+g/2)}\,\mathbb E_J\Big[\prod_{i:\sigma_i=1}J_{S_i}\Big],\\
    a_m&=\frac{i^m}{m!}\sum_{\sigma,S_1,\ldots,S_m}w_\sigma\,
    \Tr\big(\rho(0)\operatorname{ad}^{\sigma_1}_{S_1}\cdots\operatorname{ad}^{\sigma_m}_{S_m}(O)\big).
  \end{aligned}
\end{equation}
At $m=0$, the tuple sum contains the single empty tuple, scalar empty products equal $1$, and the empty composition of commutators is the identity map; hence $a_0=\Tr(\rho(0)O)$.

Most tuples do not contribute: the nested commutator is nonzero only if the sets $S_1,\ldots,S_m$ are linked to $O$ through overlapping supports.
Organising the sum by which tuples survive is the cluster expansion.

\paragraph{Clusters}
A multiset is an unordered collection of elements with repetitions; for a multiset $\mathbf W$ we write $\mathbf W^\partial$ for its underlying set, $\mu_{\mathbf W}(u)$ for the multiplicity of $u$, and $\mathbf W!\coloneqq\prod_{u\in\mathbf W^\partial}\mu_{\mathbf W}(u)!$.
The \emph{interaction graph} for the sets $S$ and an observable $O$ has one vertex for every set $S$ and one vertex for $O$, two vertices being adjacent if and only if their supports intersect.
A \emph{cluster} is a nonempty multiset $\mathbf W$ of vertices; it is \emph{connected} if the subgraph induced by $\mathbf W^\partial$ is connected, and a cluster of Hamiltonian vertices is \emph{completely connected to $O$} if $\mathbf W\cup\{O\}$ is connected.
We write $\mathcal G^O_m$ for the set of clusters of $m$ Hamiltonian vertices that are completely connected to $O$.
For empty tuples, use the convention $\mathcal G^O_0=\{\emptyset\}$.
For the models of \cref{eq:HJ} and an observable on at most $k_O$ sites, the interaction graph has maximum degree at most $\kappa N^{k-1}$ with $\kappa=\max\{k,k_O\}$: every site lies in $\binom{N-1}{k-1}\le N^{k-1}$ sets $S$, so a set $S$ intersects at most $k\binom{N-1}{k-1}-1$ other sets and possibly the observable vertex, while the observable vertex meets at most $k_O\binom{N-1}{k-1}$ sets.

\begin{lemma}[{\cite[Lemma~2]{wild2023}}]\label{lem:connected}
  If the multiset $\{S_1,\ldots,S_m\}$ is not completely connected to $O$, then
$\operatorname{ad}^{\sigma_1}_{S_1}\cdots\operatorname{ad}^{\sigma_m}_{S_m}(O)=0$ for all $\sigma\in\{0,1\}^m$.
\end{lemma}

For $m\ge1$, only tuples whose multiset $\mathbf W$ belongs to $\mathcal G_m^O$ contribute to \cref{eq:taylor}; each such multiset has $m!/\mathbf W!$ distinct orderings.

\begin{lemma}[Cluster count; sharpening of Proposition 3.6 of \citet{haah2021}]\label{lem:clusters}
  Consider a graph $(V,\mathcal E)$ of maximum degree at most $\Delta$, for some $\Delta\ge2$. For any vertex $v\in V$ and any $w\ge1$, the number of connected clusters $\mathbf W$ with $v\in\mathbf W$ and $|\mathbf W|=w$ is at most $(1+e(\Delta-1))^{w-1}\le(e\Delta)^{w-1}$.
  In particular, if the interaction graph of $H$ and $O$ has maximum degree at most $\Delta$, for some $\Delta\ge2$, then $|\mathcal G^O_m|\le(1+e(\Delta-1))^m\le(e\Delta)^m$.
\end{lemma}

The proof is given in \cref{app:cluster-count}.

\paragraph{Technical cluster expansion lemma}
For Gaussian couplings we will need sums of nested commutators in which some of the labels are forced to coincide.
For $2p\le m$, let $\mathcal P^m_{2p}$ denote the set of partial pairings of $\{1,\ldots,m\}$ consisting of $p$ disjoint unordered pairs, so that
\[
  |\mathcal P^m_{2p}|
  =\binom m{2p}(2p-1)!!
  =\frac{m!}{2^pp!(m-2p)!}.
\]
In particular, $\mathcal P^m_0=\{\emptyset\}$, with $(-1)!!=1$.
A pairing $P$ imposes equality of the corresponding labels through the product
$\prod_{\{i,j\}\in P}\delta_{S_i,S_j}$, whose value is $1$ for $P=\emptyset$.

\begin{lemma}[Nested-commutator sums]\label{lem:phiP}
  Suppose the interaction graph for the sets $S$ and $O$ has maximum degree at most $\Delta$, for some $\Delta\ge2$, and $\|h^\sigma_S\|\le1$ for all $S$ and $\sigma$.
  Then for all $m\ge2p$, $P\in\mathcal P^m_{2p}$ and $\sigma\in\{0,1\}^m$,
  \[
    \begin{aligned}
      &\sum_{S_1,\ldots,S_m}
      \big\|\operatorname{ad}^{\sigma_1}_{S_1}\cdots\operatorname{ad}^{\sigma_m}_{S_m}(O)\big\|
      \prod_{\{i,j\}\in P}\delta_{S_i,S_j} \\
      &\quad\le(m-p)!\,2^m(e\Delta)^{m-p}\|O\|.
    \end{aligned}
  \]
\end{lemma}

\begin{proof}
Due to the Kronecker-delta product, only $m-p$ of the labels are independent: each of the $p$ pairs fuses two labels into one shared value, while the $m-2p$ unpaired labels are free.
Every term satisfies $\|\operatorname{ad}^{\sigma_1}_{S_1}\cdots\operatorname{ad}^{\sigma_m}_{S_m}(O)\|\le2^m\|O\|$, and by \cref{lem:connected} only tuples whose multiset is completely connected to $O$ contribute.
A surviving tuple is determined by the multiset $\mathbf W$ of its $m-p$ independent values, which has the same underlying set as $\{S_1,\ldots,S_m\}$ and is therefore itself completely connected to $O$, i.e.\ $\mathbf W\in\mathcal G^O_{m-p}$; each such $\mathbf W$ arises from $(m-p)!/\mathbf W!\le(m-p)!$ orderings.
Hence the sum is at most $|\mathcal G^O_{m-p}|\,(m-p)!\,2^m\|O\|$, and \cref{lem:clusters} gives the claim. The case $m=p=0$ follows directly from the empty-tuple convention.
\end{proof}

\subsubsection{Analyticity of \texorpdfstring{$f(s)$}{f(s)}}
\label{sec:analyticity}

The cluster expansion bounds $f(s)$ on a region around the real axis. Let
\begin{equation}\label{eq:regions}
  \mathcal R\coloneqq
  \big\{x+iy\in\mathbb C\,\big|\,|y|(1+|x|)\le\tfrac14\big\},
\end{equation}
a region that narrows like $1/|x|$ away from the origin, see \cref{fig:regions}.

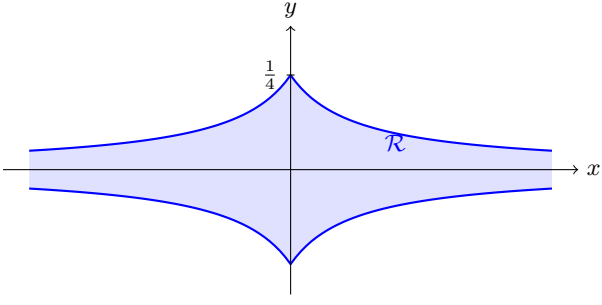
\begin{figure}[t]
  \centering
  \begin{tikzpicture}[x=0.1\columnwidth,y=5cm]
    \fill[blue!12] plot[domain=-4:4,samples=101] (\x,{0.25/(1+abs(\x))}) -- plot[domain=4:-4,samples=101] (\x,{-0.25/(1+abs(\x))}) -- cycle;
    \draw[blue,thick] plot[domain=-4:4,samples=101] (\x,{0.25/(1+abs(\x))});
    \draw[blue,thick] plot[domain=-4:4,samples=101] (\x,{-0.25/(1+abs(\x))});
    \draw[->] (-4.4,0) -- (4.4,0) node[right] {$x$};
    \draw[->] (0,-0.33) -- (0,0.38) node[above] {$y$};
    \node[blue] at (1.6,0.07) {$\mathcal R$};
    \draw (0.06,0.25) -- (-0.06,0.25) node[left] {$\tfrac14$};
  \end{tikzpicture}
  \caption{The region $\mathcal R$ (shaded) in the dimensionless coordinates $x+iy$ of \cref{eq:regions}. The corresponding scaled region on which $f(s)$ is bounded is given in \cref{lem:analyticity}.}
  \label{fig:regions}
\end{figure}

\begin{lemma}[Analyticity of $f(s)$]\label{lem:analyticity}
  Let $H$ be of the form \cref{eq:HJ} with $\|h^\sigma_S\|\le1$, and let $\rho(0)$, $O$ and the $J_S$ be as in \cref{thm:classical}. Then $f(s)$ of \cref{eq:Os} is an entire function of $s\in\mathbb C$ with $|f(s)|\le2\|O\|$ whenever $2e\kappa s\in\mathcal R$, where $\kappa=\max\{k,k_O\}$.
\end{lemma}

\begin{proof}
\emph{Setup.} Write $s=t+i\tau$ with $t,\tau\in\mathbb R$ and let $\mathcal U_t(Y)\coloneqq e^{iHt}Ye^{-iHt}$ denote the real-time Heisenberg evolution, which is an isometry of the operator norm.
Since $H$ commutes with itself, \cref{eq:hadamard} at imaginary time $i\tau$ gives
\begin{align*}
  e^{iHs}Oe^{-iHs}
  &=\mathcal U_t(e^{-H\tau}Oe^{H\tau})=\sum_{m\ge0}\frac{(-\tau)^m}{m!}
  \mathcal U_t(\operatorname{ad}_H^m(O)),
\end{align*}
with $\operatorname{ad}_H^m(O)$ expanded as in \cref{eq:ad-expansion}.
For every fixed $J$ the left-hand side is a matrix-valued entire function of $s$, bounded on compact sets by $\|O\|e^{2|\tau|\|H\|}$; since $\|H\|\le N+N^{-(k-1)/2}\sum_{|S|=k}|J_S|$ has finite exponential moments, the fixed, disorder-independent inputs imply that $f(s)$ is entire and the expectation may be taken term by term.

\emph{Bound.} Fix types $\sigma\in\{0,1\}^m$ with $d$ deterministic and $g$ Gaussian entries, and let $G\coloneqq\{i\mid\sigma_i=1\}$.
For a tuple $\boldsymbol S=(S_1,\ldots,S_m)$, the contribution of $\sigma$ to the $m$-th term is
\[
  \begin{gathered}
    \frac{(-\tau)^m}{m!\,N^{(k-1)(d+g/2)}}
    \sum_{S_1,\ldots,S_m}
    \mathbb E_J\left[\prod_{i\in G}J_{S_i}\,F_{\boldsymbol S}(J)\right],\\
    F_{\boldsymbol S}(J)\coloneqq
    \Tr\big(\rho(0)\mathcal U_t
    (\operatorname{ad}^{\sigma_1}_{S_1}\cdots\operatorname{ad}^{\sigma_m}_{S_m}(O))\big),
  \end{gathered}
\]
where $F_{\boldsymbol S}$ depends on $J$ only through $\mathcal U_t$.
We evaluate the expectation by Gaussian integration by parts, $\mathbb E[J_SG(J)]=\mathbb E[\partial_SG(J)]$ with $\partial_S\coloneqq\partial/\partial J_S$, applied to each factor $J_{S_i}$, $i\in G$, in turn: every factor is either paired with another factor $J_{S_j}$, which produces $\delta_{S_i,S_j}$, or it differentiates $F_{\boldsymbol S}$.
Summing over all ways to do so,
\begin{equation}\label{eq:ibp}
  \begin{aligned}
    \mathbb E_J\Big[\prod_{i\in G}J_{S_i}\,F_{\boldsymbol S}(J)\Big]
    &=\sum_{p=0}^{\lfloor g/2\rfloor}\sum_{P\in\mathcal P^G_{2p}}
    \Big(\prod_{\{i,j\}\in P}\delta_{S_i,S_j}\Big)\\
    &\qquad\times\mathbb E_J\Big[
    \Big(\prod_{i\in G\setminus\bigcup P}\partial_{S_i}\Big)
    F_{\boldsymbol S}(J)\Big],
  \end{aligned}
\end{equation}
where $\mathcal P^G_{2p}\subseteq\mathcal P^m_{2p}$ denotes the partial pairings of $G$ into $p$ pairs and $\bigcup P$ the set of paired indices, so $i\in G\setminus\bigcup P$ runs over the $q\coloneqq g-2p$ unpaired Gaussian indices; for a $J$-independent function this is Wick's theorem.
Write $Y_{\boldsymbol S}\coloneqq\operatorname{ad}^{\sigma_1}_{S_1}\cdots\operatorname{ad}^{\sigma_m}_{S_m}(O)$.

By \cref{eq:HJ}, $H$ is linear in each $J_S$ with $\partial_SH=N^{-(k-1)/2}h^1_S$, so $\partial_S\partial_{S'}H=0$.
Differentiating $e^{\pm iHt}$ gives Duhamel's formula
\begin{equation}\label{eq:duhamel}
  \partial_S\,\mathcal U_t(Z)
  =i\int_0^t\mathcal U_u\big([\partial_SH,\mathcal U_{t-u}(Z)]\big)\,du
\end{equation}
for $J$-independent $Z$: the derivative inserts one commutator with $\partial_SH$ at an intermediate time $u$ and splits the evolution into two real-time segments.
We apply the $q$ derivatives $\partial_{S_i}$, $i\in G\setminus\bigcup P$, one after the other.
Since neither $Y_{\boldsymbol S}$ nor the inserted commutators depend on $J$, each derivative acts through \cref{eq:duhamel} on one of the evolution segments already present, and after $j$ derivatives there are $j+1$ segments.
The product rule thus produces $q!$ terms, one for each time ordering of the $q$ labelled insertions.
Each term is an integral over the ordered times $0\le u_1\le\cdots\le u_q\le t$ (reversed for $t<0$), a simplex of volume $|t|^q/q!$, whose integrand consists of $q$ commutators with operators $N^{-(k-1)/2}h^1_{S_i}$ interleaved with real-time evolutions.
Real-time conjugation preserves the operator norm and $\|[h^1_A,Z]\|\le2\|Z\|$, so every integrand is bounded by $(2/N^{(k-1)/2})^q\|Y_{\boldsymbol S}\|$ uniformly in $J$.
The $q!$ terms and the simplex volume $|t|^q/q!$ combine to
\[
  \big\|
  \big[\big(\prod_{i\in G\setminus\bigcup P}\partial_{S_i}\big)
  \mathcal U_t\big](Y_{\boldsymbol S})\big\|
  \le\left(\frac{2|t|}{N^{(k-1)/2}}\right)^q
  \|Y_{\boldsymbol S}\|.
\]
If several unpaired indices carry the same set $S$, each derivative is still a separate labelled insertion, so the count $q!$ is unchanged; for $q=0$ the bound is the isometry property itself.

Set $x=2e\kappa|t|$ and $y=2e\kappa|\tau|$.
Using $|\Tr(\rho(0)Z)|\le\|Z\|$, \cref{lem:phiP} with $\Delta=\kappa N^{k-1}\ge2$ and types $\sigma$, the $\binom md$ choices of $\sigma$, and $|\mathcal P^G_{2p}|=g!/(2^pp!(g-2p)!)$, the $m$-th term is therefore bounded in absolute value by
\begin{align*}
  &\sum_{d+g=m}\binom md\frac{|\tau|^m}{m!\,N^{(k-1)(d+g/2)}}
  \sum_{p=0}^{\lfloor g/2\rfloor}
  |\mathcal P^G_{2p}|
  \left(\frac{2|t|}{N^{(k-1)/2}}\right)^{g-2p}\\
  &\qquad\times(m-p)!\,2^m(e\kappa N^{k-1})^{m-p}\|O\|\\
  &\quad\le\sum_{d+2p+q=m}
  \frac{(d+p+q)!}{d!\,p!\,q!}\,
  y^d\,(y^2/2)^p\,(xy)^q\,\|O\|.
\end{align*}
Here the factorials give the multinomial coefficient times $2^{-p}$; together with this $2^{-p}$, the interaction-strength factors are at most $y$ per deterministic index, $y^2/2$ per pair and $xy$ per unpaired Gaussian index, and the powers of $N$ cancel: $(N^{k-1})^{-d-g/2-q/2+m-p}=1$.
On the prescribed region, $y(1+x)\le1/4$, so $y+xy+y^2/2\le1/4+1/32=9/32$.
The nonnegative majorant can therefore be summed exactly: regrouping by $n=d+p+q=m-p$ and applying the multinomial theorem gives
\[
  \begin{aligned}
    &\sum_{m\ge0}\sum_{d+2p+q=m}
      \frac{(d+p+q)!}{d!\,p!\,q!}\,y^d(y^2/2)^p(xy)^q\\
    &\quad=\sum_{n\ge0}(y+xy+y^2/2)^n
      =\frac{1}{1-y-xy-y^2/2}.
  \end{aligned}
\]
Thus $|f(t+i\tau)|\le\|O\|/(1-y-xy-y^2/2)\le2\|O\|$.
Without deterministic terms ($d=0$) this is the binomial theorem; the deterministic terms only add $y$ to the ratio of the geometric series, which the $1$ in the condition $y(1+x)\le1/4$ defining $\mathcal R$ accommodates.
\end{proof}

\subsubsection{Analytic Continuation \& Truncated Series Evaluation}
\label{sec:algorithm}

For the continuation argument, let $f$ denote an arbitrary entire function bounded on $\gamma^{-1}\mathcal R$, with a generic scale $\gamma>0$.
For $\gamma t$ small, $f(t)$ can be computed from the truncated Maclaurin series of $f$ itself, since the disc of radius $\gamma^{-1}/5$ lies inside this region; this is analogous to the short-time algorithm of \citet{wild2023}.
To reach arbitrary $t>0$ we follow their Theorem 6 and compose $f$ with an analytic map $\phi$ satisfying $\phi(0)=0$ and $\phi(1)=1$, such that $t\phi$ maps a disc $\mathcal D_R=\{z\in\mathbb C\mid|z|\le R\}$ of radius $R>1$ into the region on which $f$ is bounded: the Maclaurin series of $g(z)=f(t\phi(z))$ then converges at $z=1$ at a rate set by $R$.

\begin{lemma}[Analytic continuation]\label{lem:continuation}
  Let $f:\mathbb C\to\mathbb C$ be entire and $\gamma,c>0$, and suppose that $|f(s)|\le c$ whenever $\gamma s\in\mathcal R$.
  For $t>0$ and $\epsilon\in(0,1)$, define
  \[
    \alpha=\exp(2\pi\gamma t(1+2\gamma t)),\qquad
    \phi(z)=-\frac{\log(1-(1-\alpha^{-1})z)}{\log\alpha},
  \]
  using the principal branch of the logarithm.
  Then $g(z)\coloneqq f(t\phi(z))$ satisfies
  \[
    \left|f(t)-\sum_{j=0}^{M}\frac{g^{(j)}(0)}{j!}\right|
    \le\epsilon c,\qquad
    M=\left\lceil(\alpha+1)\log\frac{\alpha}{\epsilon}\right\rceil.
  \]
  The map has $\phi(0)=0$, $\phi(1)=1$, and explicit Taylor coefficients
  \[
    \phi(z)=\sum_{n\ge1}\phi_nz^n,\qquad
    \phi_n=\frac{(1-\alpha^{-1})^n}{n\log\alpha},
    \qquad0<\phi_n\le1.
  \]
\end{lemma}

\begin{proof}
We use the logarithmic map from the proof of Theorem 6 of \citet{wild2023}, written here in terms of $\alpha$.
Their parametrization is $\phi(z)=\log(1-z/r)/\log(1-1/r)$; setting $r=\alpha/(\alpha-1)$ gives the formula in the statement.

Set $R=1+\alpha^{-1}>1$.
For $|z|\le R$, the argument of the logarithm obeys
\[
  \begin{gathered}
    \operatorname{Re}(1-(1-\alpha^{-1})z)
    \ge\alpha^{-2}>0,\\
    \alpha^{-2}\le\big|1-(1-\alpha^{-1})z\big|
    \le2-\alpha^{-2}\le\alpha^2.
  \end{gathered}
\]
The last inequality follows from $\alpha^2+\alpha^{-2}\ge2$.
Thus $\phi$ and $g$ are analytic on a neighbourhood of $\mathcal D_R$, with $g(1)=f(t)$.
The modulus bounds and the fact that the logarithm's argument lies in the right half-plane give, respectively,
\begin{align}
  |\operatorname{Re}\phi(z)|&\le2,\label{eq:phi-re}\\
  |\operatorname{Im}\phi(z)|&\le\frac{\pi}{2\log\alpha}.\label{eq:phi-im}
\end{align}
Our choice $\log\alpha=2\pi\gamma t(1+2\gamma t)$ therefore gives
\[
  |\operatorname{Re}(\gamma t\phi(z))|\le2\gamma t,\qquad
  |\operatorname{Im}(\gamma t\phi(z))|
  \le\frac1{4(1+2\gamma t)}.
\]
Their combination implies
\[
  |\operatorname{Im}(\gamma t\phi(z))|
  (1+|\operatorname{Re}(\gamma t\phi(z))|)\le1/4,
\]
so $\gamma t\phi(z)\in\mathcal R$ and $|g(z)|\le c$ throughout $\mathcal D_R$.

The standard Cauchy estimate bounds the Taylor tail of $g$ at $z=1$ by $c/(R^M(R-1))$.
Since $R-1=\alpha^{-1}$ and $\log R=\log(1+\alpha^{-1})\ge1/(\alpha+1)$, our choice of $M$ gives
\[
  \begin{aligned}
    \left|f(t)-\sum_{j=0}^{M}\frac{g^{(j)}(0)}{j!}\right|
    \le\frac{\alpha c}{R^M}
    \le\alpha c\exp\!\left(-\frac{M}{\alpha+1}\right)
    \le\epsilon c.
  \end{aligned}
\]
Finally, expanding the logarithm gives the stated coefficients of $\phi$.
The inequality $\log\alpha\ge1-\alpha^{-1}$ implies
$0<\phi_n\le(1-\alpha^{-1})^{n-1}/n\le1$.
\end{proof}

The truncated series of \cref{lem:continuation} is a finite sum of explicit terms; we now bound the time and space needed to enumerate and accumulate them.

\begin{lemma}[Evaluating the truncated series]\label{lem:eval}
  In the setting of \cref{thm:classical}, let $f(s)$ be as in \cref{eq:Os}.
  Let $\phi(z)=\sum_{n\ge1}\phi_nz^n$ with $|\phi_n|\le1$, the $\phi_n$ available as fixed-point numbers, and let $t,M\le\mathrm{poly}(N)$.
  Then the order-$M$ Maclaurin truncation of $g(z)=f(t\phi(z))$ at $z=1$ can be computed to absolute precision $\delta\|O\|$ in time $(\chi N)^{\mathcal O(kM+k_O)}\mathrm{polylog}(1/\delta)$ and, simultaneously, in workspace $\mathcal O((kM+k_O)\log(\chi N)+\log(1/\delta))$ bits.
\end{lemma}

The proof is given in \cref{sec:proofs}.

\begin{proof}[Proof of \cref{thm:classical}]
By \cref{eq:absorb}, after replacing $H$ by $H/((k+1)^k\mathcal J)$ and $t$ by $t'\coloneqq(k+1)^k\mathcal Jt$, $H$ is of the form \cref{eq:HJ} with $\|h^\sigma_S\|\le1$.
By \cref{lem:analyticity}, $f(s)$ is entire and bounded by $2\|O\|$ whenever $2e\kappa s\in\mathcal R$.
For $t'>0$, applying \cref{lem:continuation} at time $t'$ with $\gamma=2e\kappa$ and accuracy parameter $\epsilon/4$ gives truncation error at most $\epsilon\|O\|/2$ at order $M=\lceil(\alpha+1)\log(4\alpha/\epsilon)\rceil$.
Since $\log\alpha=2\pi\gamma t'(1+2\gamma t')=\mathcal O(1+\kappa^2t'^2)$ with absolute constants, this is $M=e^{\mathcal O(1+\kappa^2t'^2)}\log(2/\epsilon)$, as claimed, and since $M\ge\alpha\ge e^{4\pi\gamma^2t'^2}$, $M\le\mathrm{poly}(N)$ also gives $t'\le\mathrm{poly}(N)$.
\Cref{lem:eval} computes this truncation with an additional error at most $\epsilon\|O\|/2$.
Since $\log(1/\epsilon)=\mathcal O(M)$, its time and space bounds reduce to those claimed.
\end{proof}

\subsection{Reduction of Disorder Ensemble for \texorpdfstring{$r=\mathcal{O}(1)$}{r = O(1)}}
\label{sec:reduction-fixed-rank}
The ensemble-averaged dynamics estimation problem, \cref{prob:averaged-dynamics}, for $H_J$ with $r=\mathcal{O}(1)$ can always be reduced to the dynamics estimation problem without disorder of \cref{eq:task-fixed}, i.e.\ with an $H_{J=0}$, on an enlarged system.
The core idea consists of replacing the classical random disorder parameters, i.e.\ the Rademacher-distributed $v_{i\ell}$, with Pauli operators on newly introduced ancilla qubits initialised in a product state.

Each system qubit together with its $r$ ancillas forms a qudit of local dimension $\chi=2^{r+1}$.
On qudit $i$, $\sigma_i^\alpha$ acts on the system qubit only, and $Z_{i\ell}=\sigma_i^z\otimes\sigma^z_{i\ell}$ acts as $\sigma^z$ on the system qubit and on the $\ell$-th ancilla.
We write $\overline\rho(t)=\mathbb E_J[\rho(t)]$ for the disorder-averaged state of \cref{eq:task-averaged}.

\begin{restatedlemma}{thm:informal-reduction}[Exact reduction of fixed-rank disorder]
\label{lem:reduction}
Let $H^\prime_{J=0}$ be obtained from $H_J$ by replacing every $v_{i\ell}$ with $\sigma^z_{i\ell}$, let $\rho^\prime(0)=\rho(0)\otimes(\ket+\bra+)^{\otimes rN}$ and $\rho^\prime(t)=e^{-iH^\prime_{J=0}t}\rho^\prime(0)e^{iH^\prime_{J=0}t}$.
Then $\tr_{\mathrm{anc}}\rho^\prime(t)=\overline\rho(t)$ for all $t\ge0$, and hence $\overline F_O(t)=\tr\big((O\otimes I)\rho^\prime(t)\big)$ for every observable $O$ on the system qubits.
Moreover, on the $N$ qudits, $H^\prime_{J=0}$ is non-disordered, $k$-local and permutation invariant, and if $\rho(0)$ is a product state or a symmetric product state, then so is $\rho^\prime(0)$.
\end{restatedlemma}

\begin{proof}
For a sign configuration $v\in\{\pm1\}^{N\times r}$, let $H_v$ be $H_J$ with these signs and $\ket v$ the ancilla basis state with $\sigma^z_{i\ell}\ket v=v_{i\ell}\ket v$.
As the $\sigma^z_{i\ell}$ commute with all system operators, $H^\prime_{J=0}=\sum_vH_v\otimes\ket v\bra v$, and with $\ket+^{\otimes rN}=2^{-rN/2}\sum_v\ket v$,
\begin{equation*}
  \rho^\prime(t)=2^{-rN}\sum_{v,v^\prime}e^{-iH_vt}\rho(0)e^{iH_{v^\prime}t}\otimes\ket v\bra{v^\prime}.
\end{equation*}
The partial trace keeps the blocks $v=v^\prime$, i.e.\ the uniform average of $e^{-iH_vt}\rho(0)e^{iH_vt}$ over $v$, which is $\overline\rho(t)$.
Replacing the signs in a term of \cref{eq:hamiltonian} gives $\prod_{j=1}^k\sigma^z_{i_j}\sigma^z_{i_j\ell}=\prod_{j=1}^kZ_{i_j\ell}$, which acts on the qudits $i_1,\ldots,i_k$ with a site-independent coefficient, while $Nf(m_x,m_y,m_z)$ acts trivially on the ancillas; this gives $k$-locality and permutation invariance.
The statement on $\rho^\prime(0)$ is immediate.
\end{proof}

To apply the results of \cref{sec:compressed-simulation,sec:classical-sampling,sec:cluster-method}, we write $H^\prime_{J=0}$ in collective operators.
Each pattern $\ell$ contributes the elementary symmetric polynomial $e_k=\sum_{i_1<\cdots<i_k}Z_{i_1\ell}\cdots Z_{i_k\ell}$ of the commuting operators $Z_{i\ell}$.
Newton's identities express it through the power sums $p_j=\sum_iZ_{i\ell}^j$,
\begin{equation*}
  ke_k=\sum_{j=1}^k(-1)^{j-1}e_{k-j}\,p_j,\qquad e_0=I,
\end{equation*}
e.g.\ $2e_2=p_1^2-p_2$, and since $Z_{i\ell}^2=I$, these reduce to $p_j=\sum_iZ_{i\ell}=Nm(Z_\ell)$ for odd $j$ and $p_j=NI$ for even $j$.
Thus $e_k=N^kE_k(m(Z_\ell))$ with a polynomial $E_k$ of degree $k$ and, by induction, coefficient norm $\|E_k\|_1\le1$; for instance, $E_2(m)=\frac12(m^2-1/N)$.
The algorithms below compute the coefficients of $E_k$ themselves, and can do so exactly in small space.
Written as a polynomial in $Nm(Z_\ell)=\sum_iZ_{i\ell}$, the recursion only multiplies by $Nm(Z_\ell)$ or by the integer $N$ and divides by $k$, so by induction $k!\,e_k$ has integer coefficients, which are at most $k!N^k$ in modulus since $\|E_k\|_1\le1$.
Every coefficient of $E_0,\ldots,E_k$ is thus a ratio of integers of $\mathcal O(k\log N)$ bits, and running the recursion, which keeps all $\mathcal O(k^2)$ of them, takes polynomial time and $\mathcal O(k^3\log N)$ space, i.e.\ $\mathcal O(\log N)$ for $k$ independent of $N$.
Hence
\begin{equation}
  H^\prime_{J=0}=N\Big(f(m_x,m_y,m_z)-\frac J{\sqrt r}\sum_{\ell=1}^rE_k\big(m(Z_\ell)\big)\Big),
  \label{eq:fixed-rank-dilated}
\end{equation}
with $m_\alpha=m(\sigma^\alpha)$, is of the form \cref{eq:compression-normalized-family} with the unit-norm operators $\sigma^x_i,\sigma^y_i,\sigma^z_i,Z_{i1},\ldots,Z_{ir}$ on each qudit and coefficient norm at most
\begin{equation}
  \Lambda\coloneqq\|f\|_1+|J|\sqrt r.
  \label{eq:fixed-rank-coefficient-norm}
\end{equation}

The proof of \cref{lem:reduction} only uses that $H^\prime_{J=0}$ is obtained by replacing the signs $v_{i\ell}$ with $\sigma^z_{i\ell}$, not the values of the coefficients.
It therefore applies verbatim to the quantum Hopfield models \cref{eq:models-hopfield} with arbitrary real weights $\mu_\ell$, giving coefficient norm $\Lambda=|B|+\sum_\ell|\mu_\ell|$.

We now state the formal versions of \cref{thm:hopfield-quantum,thm:hopfield-sampling,thm:hopfield-classical} for $H_J$, with $k$ and $k_O$ independent of $N$.

\begin{restatedtheorem}{thm:hopfield-quantum}[Quantum simulation of fixed-rank disordered dynamics]
\label{thm:hopfield-quantum-formal}
Let $\Lambda\le\mathrm{poly}(N)$, and let $\rho(0)$, $t$, $\epsilon$ and $O$ be as in \cref{thm:compression-dynamics}.
Then a quantum algorithm solves \cref{prob:averaged-dynamics} for $H_J$ using, simultaneously, $\mathrm{poly}(t,1/\epsilon)\,N^{\mathcal O(2^r)}$ time, $\mathcal O(2^r\log N)$ qubits and $\mathcal O(2^r\log N+\log(t+1/\epsilon))$ bits of classical workspace.
\end{restatedtheorem}

\begin{restatedtheorem}{thm:hopfield-sampling}[Classical sampling of fixed-rank disordered dynamics at short times]
\label{thm:hopfield-sampling-formal}
Let $r$ and $\Lambda$ be independent of $N$, $t=\mathcal O(\log N/N)$, $\epsilon^{-1}\le\mathrm{poly}(N)$ and $\rho(0)=(\ket\phi\bra\phi)^{\otimes N}$ with single-site amplitudes given to inverse-polynomial precision.
Then a classical algorithm using random bits outputs, bit by bit, a string whose distribution is within total variation distance $\epsilon$ of the computational-basis measurement distribution $\overline P(x)=\bra x\overline\rho(t)\ket x$, in time $N^{\mathcal O(2^r)}$ and, simultaneously, $\mathcal O(2^r\log N)$ workspace.
\end{restatedtheorem}

\begin{restatedtheorem}{thm:hopfield-classical}[Classical simulation of fixed-rank disordered dynamics at constant times]
\label{thm:hopfield-classical-formal}
Let $\Lambda$ be independent of $N$, and let $\rho(0)=\bigotimes_{i=1}^N\rho_i$ be a product state and $O$ an observable supported on at most $k_O$ qubits.
Then $\overline F_O(t)$ for $H_J$ can be computed to accuracy $\epsilon\|O\|$ by a deterministic classical algorithm in time $(2^rN)^{\mathcal O(M)}$ and, simultaneously, workspace $\mathcal O(M(r+\log N))$, with $M=e^{\mathcal O(1+\Lambda^2t^2)}\log(2/\epsilon)$, under the conditions of \cref{thm:classical}.
In particular, for $r,t=\mathcal O(1)$ and $\epsilon=\Theta(1)$, this is $\mathrm{poly}(N)$ time and $\mathcal O(\log N)$ workspace.
\end{restatedtheorem}

\begin{proof}[Proof of \cref{thm:hopfield-quantum-formal,thm:hopfield-sampling-formal,thm:hopfield-classical-formal}]
With \cref{lem:reduction},
we apply \cref{thm:compression-dynamics,thm:classical-sampling,thm:classical}, respectively, and substituting $\chi=2^{r+1}$ in their bounds gives the stated ones.
Their input assumptions hold as follows.
By \cref{eq:fixed-rank-dilated}, $H^\prime_{J=0}=Nf^\prime$ is supplied as in \cref{eq:compression-normalized-family} with $\|f^\prime\|_1\le\Lambda$, every entry and coefficient being computed in $\mathcal O(r+\log N)$ space whenever it is read; $O\otimes I$ is supported on the qudits of the at most $k_O$ sites of $O$, and its matrix entries are those of $O$ times Kronecker deltas in the ancilla indices.
For $\rho(0)=(\ket\phi\bra\phi)^{\otimes N}$, $\rho^\prime(0)$ is the symmetric product state of $\ket{\phi^\prime}=\ket\phi\otimes\ket+^{\otimes r}$, whose amplitudes are those of $\ket\phi$ times $2^{-r/2}$; its compressed amplitudes $M_N(\mathbf n)^{1/2}\prod_b(\phi^\prime_b)^{n_b}$, cf.\ \cref{eq:generalised-dicke-states}, are computed to inverse-polynomial precision in polynomial time and $\mathcal O(\chi\log N)$ space.
For \cref{thm:classical}, $H^\prime_{J=0}$ is, up to a constant, of the form \cref{eq:classical-hamiltonian} with $h'_T=0$ and $\|h_T\|\le k^k\Lambda$, as shown after \cref{thm:classical}, so \cref{thm:classical} with $\mathcal J=k^k\Lambda$ gives the stated $M$; the matrix elements of $\rho_i\otimes(\ket+\bra+)^{\otimes r}$ are those of $\rho_i$ times $2^{-r}$.

For sampling, a computational-basis measurement of $\rho^\prime(t)$ yields the system bits $x$ and the ancilla bits, and by \cref{lem:reduction} the marginal of $x$ is $\overline P(x)$.
The sampler of \cref{thm:classical-sampling} outputs the qudit symbols one at a time, so the ancilla bits are dropped as they are produced; as marginalisation does not increase the total variation distance, the system bits are sampled to accuracy $\epsilon$.
\end{proof}

\subsection{Disorder Removal through Gauge Transformation for \texorpdfstring{$H_J$}{H\_J} at \texorpdfstring{$r=1$}{r = 1}}
\label{app:rank-one-gauge}
We consider the case of $r=1$ of the Hopfield model of \cref{eq:models-hopfield}, and discuss how the model can be mapped to a non-disordered Lipkin-Meshkov-Glick (LMG) model through a gauge transformation.

Let $X_v=\prod_i(\sigma_i^x)^{(1-v_{i1})/2}$, which flips exactly the sites with $v_{i1}=-1$. Then $X_v\sigma_i^zX_v=v_{i1}\sigma_i^z$ and $X_v\sigma_i^xX_v=\sigma_i^x$, hence
\begin{align}
	H_{\mathrm{Hop}}  =\frac{\mu_1}{N}\Big(\sum_iv_{i1}\sigma_i^z\Big)^2+BNm_x=X_vH_{\mathrm{LMG}}X_v
	\label{eq:r1-gauge}
\end{align}
with 
	$H_{\mathrm{LMG}}  =N\big(\mu_1m_z^2+Bm_x\big)$.
The disorder thus only conjugates the clean Lipkin--Meshkov--Glick dynamics, $e^{-itH_{\mathrm{Hop}}}=X_ve^{-itH_{\mathrm{LMG}}}X_v$, and its effect depends on the initial state and observable.

Since $X_v\ket{+^N}=\ket{+^N}$ and $X_vm_xX_v=m_x$,
\begin{align}
	 & \bra{+^N}e^{itH_{\mathrm{Hop}}}m_xe^{-itH_{\mathrm{Hop}}}\ket{+^N}
	\nonumber                                                              \\
	 & \quad=\bra{+^N}e^{itH_{\mathrm{LMG}}}m_xe^{-itH_{\mathrm{LMG}}}\ket{+^N}
	\label{eq:r1-gauge-mx}
\end{align}
for every pattern, i.e.\ exactly the clean dynamics.
In contrast, $X_v\ket{0^N}=\ket{x}$ is the computational basis state with $(-1)^{x_i}=v_{i1}$, and $X_vm_zX_v=N^{-1}\sum_iv_{i1}\sigma_i^z$. With $\sigma_i^z(t)=e^{itH_{\mathrm{LMG}}}\sigma_i^ze^{-itH_{\mathrm{LMG}}}$ and $v_{i1}\ket{x}=\sigma_i^z\ket{x}$, averaging uniformly over the pattern $(v_{i1})_i$ gives
\begin{align}
	 & \overline{\bra{0^N}e^{itH_{\mathrm{Hop}}}m_ze^{-itH_{\mathrm{Hop}}}\ket{0^N}}
	\nonumber                                                              \\
	 & \quad=\frac{1}{N2^N}\sum_{i}\sum_{x}\bra{x}\sigma_i^z(t)\sigma_i^z\ket{x}
	\nonumber                                                              \\
	 & \quad=2^{-N}\tr\big(\sigma_1^z(t)\sigma_1^z\big),
	\label{eq:r1-gauge-mz}
\end{align}
using permutation symmetry of $H_{\mathrm{LMG}}$ in the last step. This is the clean infinite-temperature local autocorrelation, whose trace runs over the full $2^N$-dimensional Hilbert space rather than the permutation-symmetric sector, so the clean Hamiltonian alone does not make the disordered problem trivial.
The remaining two combinations ($\langle +^N|m_z(t)|+^N\rangle$, $\langle 0^N|m_x(t)|0^N\rangle$) vanish, the second only after averaging. For $\ket{+^N}$, the expectation of $m_z$ is $N^{-1}\sum_iv_{i1}$ times its clean value, which vanishes because $H_{\mathrm{LMG}}$ and $\ket{+^N}$ are invariant under $\prod_i\sigma_i^x$, which maps $m_z$ to $-m_z$. For $\ket{0^N}$, $m_x$ follows the clean dynamics started from $\ket{x}$, whose pattern average is $2^{-N}\tr(m_x)=0$.

At higher rank, removing $v_{i1}$ changes each remaining pattern to $v_{i1}v_{i\ell}$. These relative patterns generally remain site dependent, so the same gauge cannot make all interaction terms clean.

The clean dynamics in \cref{eq:r1-gauge-mx} can be non-trivial on long time scales. For instance, at $B=0$ write $H_{\mathrm{LMG}}=(\mu_1/N)(1+2\sigma_1^zS+S^2)$ with $S=\sum_{j\ne1}\sigma_j^z$. Then $\sigma_1^x(t)=\sigma_1^x\cos(4\mu_1tS/N)-\sigma_1^y\sin(4\mu_1tS/N)$, and since the $\sigma_j^z$ are independent and uniformly $\pm1$ in $\ket{+^N}$,
\begin{equation}
	\bra{+^N}e^{itH_{\mathrm{LMG}}}m_xe^{-itH_{\mathrm{LMG}}}\ket{+^N}=\cos^{N-1}\Big(\frac{4\mu_1t}{N}\Big).
	\label{eq:r1-gauge-b0}
\end{equation}
This revives to $(-1)^{N-1}$ at $t=\pi N/(4\mu_1)$.
The same argument applies at any rank, since at $B=0$ the Hamiltonian \cref{eq:models-hopfield} is diagonal in the computational basis. With $K_{ij}=\sum_\ell\mu_\ell v_{i\ell}v_{j\ell}$, the terms containing $\sigma_i^z$ are $(2/N)\sigma_i^z\sum_{j\ne i}K_{ij}\sigma_j^z$ up to a constant, so
\begin{equation}
	\bra{+^N}e^{itH_{\mathrm{Hop}}}m_xe^{-itH_{\mathrm{Hop}}}\ket{+^N}=\frac1N\sum_i\prod_{j\ne i}\cos\Big(\frac{4K_{ij}t}{N}\Big).
	\label{eq:any-r-b0}
\end{equation}

\section{Proofs}
\label{sec:proofs}
This section collects the longer proofs of \cref{sec:technical-discussion}; the statements and the surrounding notation are those of the respective subsections.

\subsection{Proofs for \cref{sec:hardness}}

\begin{proof}[Proof of \cref{lem:circuit-to-control}]
	\Cref{lem:pairing} for each gate and \cref{lem:routing} for each of its $d/2$ pairs give $VU=\widetilde UV$, where $\widetilde U$ is a product of at most $n^{\mathrm{PR}}=Ld^2$ adjacent pair rotations $R_\alpha^{\ell,\ell+1}(\theta)$ with $\ell+1<d$ and $\theta=\lambda\pi$, $\lambda\in\{\pm1/8,\pm1/4,\pm1/2\}$. The compiler realizes each of them by the pulse of \cref{lem:selective-pulse} with $R$ as in \cref{eq:hardness-resource-choice}, so that the transverse amplitude in \cref{eq:hardness-pulse} is $\lambda/(8NR\sqrt{(\ell+1)(N-\ell)})$ and $\pi$ enters only through $\tau=8\pi NR$.

	The synthesis error of each pulse is bounded by \cref{eq:hardness-pulse-error} and $|\theta|\le\pi$,
	\begin{equation*}
		\left\|e^{-i\tau H_{\ell,\alpha,\theta}}-R_\alpha^{\ell,\ell+1}(\theta)\right\|
		\le\frac{\pi N}{2R}
		\le\frac{\pi\eta}{8n^{\mathrm{PR}}}
		<\frac{\eta}{2n^{\mathrm{PR}}}.
	\end{equation*}

	The compiler writes each amplitude and the duration with absolute error at most $\delta$, clipping amplitudes to $[-1,1]$ and keeping $h_2=1/4$ exact; call the rounded pulse $\widetilde H,\widetilde\tau$. Each of the generators $Nm_x,Nm_y,Nm_z,Nm_z^2$ in \cref{eq:hardness-controls} has norm at most $N$ since $\|m_\alpha\|\le1$, so $\|H_{\ell,\alpha,\theta}-\widetilde H\|\le4N\delta$ and $\|\widetilde H\|\le4N$. Since $\|e^{-iA}-e^{-iB}\|\le\|A-B\|$ for Hermitian $A,B$,
	\begin{align}
		\|e^{-i\tau H_{\ell,\alpha,\theta}}-e^{-i\widetilde\tau\widetilde H}\|
		 & \le\tau\|H_{\ell,\alpha,\theta}-\widetilde H\|
		+|\tau-\widetilde\tau|\,\|\widetilde H\|\nonumber \\
		 & \le4N(\tau+1)\delta
		\le\frac{\eta}{2n^{\mathrm{PR}}}
		\label{eq:hardness-rounding-bound}
	\end{align}
	for the choice
  $
		\delta=\eta/{8n^{\mathrm{PR}}N(\tau+1)}.
		\label{eq:hardness-rounding-precision}
	$

	By the triangle inequality, each emitted segment $e^{-i\widetilde\tau\widetilde H}$ is within $\eta/n^{\mathrm{PR}}$ of its target rotation $R_\alpha^{\ell,\ell+1}(\theta)$. Telescoping over the at most $n^{\mathrm{PR}}$ unitary segments gives $\|U_{\mathrm{phys}}-\widetilde U\|\le\eta$ on the symmetric subspace, and \cref{eq:hardness-intertwining} follows from $VU=\widetilde UV$.

	Since $\delta\le1$ and $R\le4Nn^{\mathrm{PR}}/\eta+1$, the total duration is at most
	\begin{equation*}
		n^{\mathrm{PR}}(\tau+\delta)\le n^{\mathrm{PR}}(8\pi NR+1)
		=\mathcal O(N^2(n^{\mathrm{PR}})^2/\eta).
	\end{equation*}

	As for the classical resources, a first pass over the input counts $L$, which fixes $R$ and $\delta$. The compiler then streams: it holds only the current gate, its mask, the current pair and route counter, a stage label, and a constant number of fixed-point registers, and writes each pulse to the output as soon as it is generated. No expanded circuit, permutation table, or state vector is stored. Since $n^{\mathrm{PR}}$, $\tau$, and $1/\delta$ are polynomial in $N,L,1/\eta$, all counters and registers have $p=\mathcal O(\log(N L / \eta))$ bits.
  Each of the at most $n^{\mathrm{PR}}=Ld^2$ segments costs a constant number of $p$-bit multiplications, divisions and one square root, i.e.\ $\mathrm{poly}(p)$ bit operations, with $\pi$ and $\tau$ computed once; enumerating the pairs of a gate adds $\mathcal O(dp)$ per gate. With $d\le N+1$ this gives time $\mathcal O(LN^2\,\mathrm{polylog}(NL/\eta))$ and logarithmic workspace.
\end{proof}

\begin{proof}[Proof of \cref{lem:selective-pulse}]
All matrices act on the symmetric subspace. Write $H_{\ell,\alpha,\theta}=Q_\ell+g\,m_\alpha,$
where $Q_\ell=\frac N4m_z^2+\frac{2\ell+1-N}2m_z$ is diagonal and $g=\theta/(\tau b_\ell)$. Substituting $m_z\ket n_N=(1-2n/N)\ket n_N$ from \cref{eq:compression-normalized-entries}, the eigenvalues $\mu_n$ of $Q_\ell$ satisfy
\begin{equation}
  \mu_{n+1}-\mu_n=\omega_n:=\frac{2(n-\ell)}N,\qquad \tau\mu_n\in2\pi\mathbb Z,
  \label{eq:hardness-pulse-spectrum}
\end{equation}
the latter because $4N\mu_n$ is an integer and $\tau=8\pi NR$. So the selected transition $\ell\leftrightarrow \ell+1$ is resonant, every other transition is detuned by $|\omega_n|\ge2/N$ and completes $\omega_n\tau/2\pi=8R(n-\ell)$ full cycles over the pulse, and $e^{-i\tau Q_\ell}=I$.

The control bounds are immediate: $h_2=1/4$, $|h_z|\le(N-1)/(2N)$, and $|h_\alpha|=|g|/N\le1/(8NR\sqrt N)$ since $N^2b_\ell^2=(\ell+1)(N-\ell)\ge N$.

Let $A=b_\ell S^\alpha_{\ell,\ell+1}$ be the resonant term of $m_\alpha=\sum_nb_nS^\alpha_{n,n+1}$, so that $g\tau b_\ell=\theta$ gives $e^{-ig\tau A}=R^{\ell,\ell+1}_\alpha(\theta)$. Pass to the interaction picture with respect to $Q_\ell$,
\begin{equation*}
  C(t)=e^{itQ_\ell}m_\alpha e^{-itQ_\ell},\qquad
  W(t)=e^{itQ_\ell}e^{-itH_{\ell,\alpha,\theta}},
\end{equation*}
so that $W'=-igCW$, $W(0)=I$, and $W(\tau)=e^{-i\tau H_{\ell,\alpha,\theta}}$. Conjugation by $e^{itQ_\ell}$ multiplies the $(n,n+1)$ entry of $m_\alpha$ by $e^{-i\omega_nt}$: the resonant entries are constant and form $A$, all others oscillate. Their integral $F(t)=\int_0^t\big(C(s)-A\big)\,ds$ has entries
\begin{equation*}
  {}_N\!\bra nF(t)\ket{n+1}_N
  ={}_N\!\bra nm_\alpha\ket{n+1}_N\frac{1-e^{-i\omega_nt}}{i\omega_n}
\end{equation*}
for $n\ne\ell$; it is Hermitian and nearest-neighbour, vanishes at $t=0$ and, by the integer cycle count, at $t=\tau$. Its entries have modulus at most $2b_n/|\omega_n|\le N$, so the row-sum bound gives $\|F(t)\|\le2N$ for all $t$.

To compare $W$ with the target, differentiate the relative evolution, $\frac{d}{dt}\big(e^{igtA}W\big)=-ige^{igtA}(C-A)W$, and integrate by parts using $C-A=F'$ and $F(0)=F(\tau)=0$:
\begin{align*}
  e^{ig\tau A}W(\tau)-I
  &=-ig\int_0^\tau e^{igtA}F'W\,dt\\
  &=g^2\int_0^\tau e^{igtA}\big(FC-AF\big)W\,dt.
\end{align*}
With $\|C(t)\|=\|m_\alpha\|\le1$, $\|A\|=b_\ell\le1$, and $e^{igtA}$, $W(t)$ unitary,
\begin{equation*}
  \|W(\tau)-e^{-ig\tau A}\|
  \le4g^2\tau N
  =\frac{\theta^2}{2\pi Rb_\ell^2}
  \le\frac{\theta^2N}{2\pi R}.
  \qedhere
\end{equation*}
\end{proof}

\subsection{Proofs for \cref{sec:cluster-method}}

\begin{proof}[Proof of \cref{lem:eval}]
Using the coefficients $a_m$ of \cref{eq:taylor}, write $g(z)=\sum_{m\ge0}t^ma_m\phi(z)^m$.
Since $\phi$ has no constant term,
\[
  [z^j]\phi(z)^m
  =\sum_{\substack{n_1+\cdots+n_m=j\\n_i\ge1}}
  \phi_{n_1}\cdots\phi_{n_m}
\]
vanishes for $m>j$, and the truncation is
\[
  \sum_{j=0}^{M}\frac{g^{(j)}(0)}{j!}
  =\sum_{j=0}^{M}\sum_{m=0}^{j}
  \sum_{\substack{n_1,\ldots,n_m\ge1\\
                  n_1+\cdots+n_m=j}}
  \phi_{n_1}\cdots\phi_{n_m}\,t^ma_m.
\]
Here $[z^j]$ denotes coefficient extraction.
For $m=0$, the inner sum is $1$ if $j=0$ and $0$ otherwise.
By Wick's theorem, the Gaussian expectation in $w_\sigma$ of \cref{eq:taylor} equals
$\prod_{S:\mu(S)>0}(\mu(S)-1)!!$ if every set $S$ occurs an even number $\mu(S)$ of times among the $S_i$ with $\sigma_i=1$, and $0$ otherwise; in particular only types $\sigma$ with an even number $g$ of Gaussian entries contribute.
Hence $0\le N^{(k-1)(d+g/2)}w_\sigma\le\prod_{S:\mu(S)>0}\mu(S)!\le m!$.

Expand the nested commutator into $2^m$ signed ordered products,
\[
  \sum_{A\subseteq\{1,\ldots,m\}}
  (-1)^{m-|A|}
  \left(\prod_{i\in A}^{\rightarrow}h^{\sigma_i}_{S_i}\right)
  O
  \left(\prod_{i\notin A}^{\leftarrow}h^{\sigma_i}_{S_i}\right),
\]
where the arrows indicate increasing and decreasing index order, respectively, and empty operator products equal $I$.
Let $e_{ab}=\ket a\bra b$, $a,b\in\{0,\ldots,\chi-1\}$, denote the matrix units of one site.
Expand each term $h^{\sigma_i}_{S_i}$ and $O$ in the tensor products of matrix units of the $k$ and at most $k_O$ sites they act on; these have at most $\chi^{2k}$ and $\chi^{2k_O}$ components, whose coefficients are matrix elements of the operator and hence bounded in modulus by its norm.
Every term then becomes a product of matrix units; on each site, a product of matrix units is again a matrix unit or zero, and the trace against the product state $\rho(0)$ factorises into single-site matrix elements $\Tr(\rho_ie_{ab})=\bra b\rho_i\ket a$, of modulus at most one.

The truncation is then an explicit sum of $M^2\cdot2^M\cdot(2N^k)^M\cdot(2\chi^{2k})^M\cdot\mathcal O(\chi^{2k_O})=(\chi N)^{\mathcal O(kM+k_O)}$ terms, each a product of $\mathcal O(M)$ factors of modulus at most $1$ or $\mathrm{poly}(N)$: the $\phi_{n_i}$, the $m$ factors of $t$, the prefactor $w_\sigma/m!\le1$, matrix elements of modulus at most $1$ or $\|O\|$, and single-site matrix elements of the $\rho_i$.

Since only even $g$ contribute, $N^{-(k-1)(d+g/2)}$ has an integer exponent. The normalization therefore requires only integer powers of $N$ and introduces no square-root approximation into coefficient evaluation. The integers $N^{(k-1)(d+g/2)}w_\sigma$, $m!$, and these powers of $N$ use $\mathcal O(kM\log N)$ bits for $M\le\mathrm{poly}(N)$.
The algorithm enumerates these terms and accumulates them.
A term is specified by an index of $\mathcal O((kM+k_O)\log(\chi N))$ bits and has modulus at most $\|O\|N^{\mathcal O(M)}$, so a fixed-point accumulator of $b=\mathcal O((kM+k_O)\log(\chi N)+\log(1/\delta))$ bits, with every product rounded to $b$ bits, reaches absolute precision $\delta\|O\|$.
The machine holds one index and one accumulator, $\mathcal O(b)$ bits, and spends $2^{\mathcal O(k)}\mathrm{poly}(kM+k_O,b)$ arithmetic operations per term, the $2^{\mathcal O(k)}$ for the matrix elements of the $\tilde h_S$ of \cref{eq:absorb}, for total time $(\chi N)^{\mathcal O(kM+k_O)}\mathrm{polylog}(1/\delta)$.
\end{proof}

\section{Outlook}
\label{sec:discussion}
An interesting direction for future work is to explore experimental implementations of quantum simulations in the settings considered here.
Such implementations could potentially demonstrate an advantage in the simultaneous time--space manner described, while requiring hardware to only provide a relatively small number of qubits.

On the theoretical side, another avenue is to extend the present dynamical results to equilibrium spin-glass physics through compressed access to disorder-averaged Gibbs and ground-state observables.
The corresponding quantum question there is whether replica methods for the disorder-dependent Gibbs normalization, combined with thermal-state preparation or spectral filtering, could yield polynomial-time algorithms using logarithmically many qubits. On the classical side, one could seek polynomial-time, polylogarithmic-space algorithms in tractable regimes, such as high temperature or weak disorder via cluster expansions, together with complexity-theoretic evidence for a quantum advantage in low-temperature or ground-state regimes.

Finally, developing new techniques for the SK and quantum $p$-spin models, and generalising ideas from the fixed-rank setting, would help address the clear physics motivation for the former.

\begin{acknowledgments}
We thank Harry Buhrman and Chinmay Nirkhe for helpful discussions.
The work is supported by the German Federal Ministry of Education and Research (BMBF) through the funded project ALMANAQC, grant number 13N17236 within the research program “Quantum Systems”.
The research is partly funded by THEQUCO as part of the Munich Quantum Valley, which is supported by the Bavarian state government with funds from the Hightech Agenda Bayern Plus. 
R.T. acknowledges support from the European Union’s Horizon Europe research and innovation program under grant agreement number 101221560 (ToNQS) and from the Munich Center for Quantum Science and Technology (MCQST), funded by the Deutsche Forschungsgemeinschaft (DFG) under Germany’s Excellence Strategy (EXC2111-390814868).
The numerical calculations used SciPy~\cite{virtanen2020scipy} and quimb~\cite{gray2018quimb}.
Astra/Fable 5.1/Opus 5.5 was used within the TeXRA multi-agent framework~\cite{lu2026multiagentautoformalizationtensornetwork} to assist with selected derivations, code development, and manuscript preparation. The authors take full responsibility for the content of the manuscript.
\end{acknowledgments}

\appendix

\section{Normalisation and Extensivity}
\label{app:normalization}
For fixed $k\ge2$, expanding the couplings in \cref{eq:hamiltonian} gives
\begin{gather}
 H_J=Nf(m_x,m_y,m_z)-\frac{J}{D}
 \sum_{\ell=1}^r\sum_{i_1<\cdots<i_k}
 \prod_{a=1}^k(v_{i_a\ell}\sigma^z_{i_a}),
 \label{eq:normalization-hamiltonian}\\
 D=\begin{cases}
 N^{k-1}\sqrt r,&1\le r<N^{k-1},\\
 \sqrt{rN^{k-1}},&r\ge N^{k-1}.
 \end{cases}
 \label{eq:normalization-denominator}
\end{gather}
The factor $\sqrt r$ comes from the coupling definition in \cref{eq:hamiltonian}, and the second case of $D$ is the replacement of the prefactor $J/N^{k-1}$ by $J/N^{(k-1)/2}$ for $r\ge N^{k-1}$ stated in \cref{sec:summary_of_results}.
For fixed $J$ and $f$ with sum of absolute coefficients $\|f\|_1$, we show
\begin{equation}
 \Pr\!\left(\|H_J\|>(\|f\|_1+C_k|J|)N\right)
 \le2e^{-c_kN}
 \label{eq:normalization-extensive}
\end{equation}
in both regimes, where $C_k$ and $c_k>0$ are constants depending only on $k$.

\begin{proof}
Since $\|m_\alpha\|\le1$, the deterministic part has norm at most $\|f\|_1N$. The disorder is diagonal, so its norm is $|J|\max_z|W(z)|/D$, where
\[
 W(z)=\sum_{\ell=1}^rP_\ell(z),\qquad
 P_\ell(z)=\sum_{i_1<\cdots<i_k}\prod_{a=1}^k v_{i_a\ell}z_{i_a}.
\]
For fixed $z\in\{-1,1\}^N$, the $P_\ell$ are independent and have mean zero.

We first control rare large contributions from a single pattern. Write $w_i=v_{i\ell}z_i$; since $z$ is fixed, these are independent uniform signs, and the pattern contribution $P_k=\sum_{i_1<\cdots<i_k}w_{i_1}\cdots w_{i_k}$ is the coefficient of $t^k$ in $F(t)=\prod_i(1+tw_i)$. The only statistic that matters is $A=\sum_iw_i$. Let $n_\pm$ be the numbers of $w_i$ equal to $\pm1$, so $A=n_+-n_-$, and suppose $A\ge0$; otherwise flip all $w_i$, which changes $P_k$ by at most a sign. Then $F(t)=(1-t^2)^{n_-}(1+t)^{A}$, and reading off the coefficient of $t^k$ and bounding binomial coefficients by powers,
 $
 |P_k|=\Bigl|\sum_{m}(-1)^m\binom{n_-}{m}\binom{A}{k-2m}\Bigr|
 \le\sum_mN^mA^{k-2m}\le(|A|+\sqrt N)^k.
$
Set $x=|A|/\sqrt N$, so $|P_k|\le N^{k/2}(1+x)^k$. Typically $x=\mathcal O(1)$ and $P_k=\mathcal O(N^{k/2})$, but $P_k$ reaches order $N^k$ when the signs align, $x\sim\sqrt N$. Such alignments are rare, since a sum of independent signs concentrates,
\[
 \Pr(x>s)\le2e^{-s^2/2},
\]
but they must be controlled after exponentiation. Restrict $|u|\le a_kN^{-(k-1)}$. Since $x\le\sqrt N$,
\begin{align}
 |uP_k|&\le a_k(1+x)^2\Bigl(\frac{1+x}{\sqrt N}\Bigr)^{k-2}\\
 &\le2^{k-2}a_k(1+x)^2\le2^{k-1}a_k(1+x^2).
\end{align}
Choose $a_k$ so that $2^{k-1}a_k\le1/4$. Then
\[
 \mathbb E[P_k^2e^{|uP_k|}]
 \le e^{1/4}N^k\,\mathbb E\bigl[(1+x)^{2k}e^{x^2/4}\bigr]
 \le M_kN^k
\]
for a constant $M_k$ depending only on $k$: integrating the tail bound, the decay $e^{-s^2/2}$ outweighs the weight $e^{s^2/4}$, so the expectation is bounded independently of $N$.

Finally, $\mathbb EP_k=0$ because each product contains distinct independent signs. The inequality $e^x\le1+x+x^2e^{|x|}/2$ gives
$\mathbb Ee^{uP_k}\le1+M_ku^2N^k\le e^{M_ku^2N^k}$, hence
\begin{equation}
 \log\mathbb Ee^{uP_\ell(z)}\le M_ku^2N^k,
 \qquad |u|\le a_kN^{-(k-1)}.
 \label{eq:normalization-pattern-moment}
\end{equation}

In both regimes, $D\ge N^{k-1}$ and $D^2\ge rN^{k-1}$. Choose $u=a_k/D$; the first inequality ensures this is within the allowed range. Independence multiplies the exponential moments, so their logarithms add:
$\log\mathbb Ee^{\pm uW(z)}\le M_kr u^2N^k\le M_ka_k^2N$.
For a fixed configuration and any $L>0$, Markov's inequality applied to $e^{\pm uW(z)}$ gives
\[
 \Pr(|W(z)|>LND)\le2e^{N(-a_kL+M_ka_k^2)}.
\]
The norm requires control of every configuration, not just one. A union bound multiplies this probability by $2^N$, giving
\[
 \Pr\!\left(\max_z|W(z)|>LND\right)
 \le2e^{N(\log2-a_kL+M_ka_k^2)}.
\]
Fix $L$ large enough that $c_k=a_kL-M_ka_k^2-\log2$ is positive and set $C_k=L$; both depend only on $k$.
Write $H_{\mathrm{dis}}=H_J-Nf(m_x,m_y,m_z)$ for the disorder part.
Since $\|H_{\mathrm{dis}}\|=|J|\max_z|W(z)|/D$, this gives $\|H_{\mathrm{dis}}\|\le C_k|J|N$ with probability at least $1-2e^{-c_kN}$. Adding the deterministic bound proves \cref{eq:normalization-extensive}.
\end{proof}

For $r<N^{k-1}$, choosing $u=a_k/N^{k-1}$ in the same argument gives the sharper bound $\|H_{\mathrm{dis}}\|\le C_k|J|N/\sqrt r$ with probability at least $1-2e^{-c_kN}$, with the same $C_k$ and $c_k$: the moment bound becomes $M_ka_k^2rN^{2-k}\le M_ka_k^2N$, and the Markov and union steps are unchanged.

\section{Permutation-Invariant Observables}
\label{app:local-observables}
\Cref{prob:averaged-dynamics} is stated for observables supported on at most $k_O$ sites.
We show that $k_O$-local permutation-invariant observables reduce to these, up to a constant factor in accuracy and with conversions in logarithmic space, because the averaged state is permutation invariant.

Indeed, the couplings $J_{i_1\cdots i_k}$ are symmetric in their indices and, the $v_{i\ell}$ being i.i.d., their joint distribution is invariant under relabelling the sites, so $U_\pi H_JU_\pi^\dagger$ has the same distribution as $H_J$.
Together with $U_\pi\rho(0)U_\pi^\dagger=\rho(0)$, this gives $U_\pi\overline\rho(t)U_\pi^\dagger=\overline\rho(t)$ for the averaged state of \cref{sec:reduction-fixed-rank} and all $\pi\in S_N$.
Hence every relabelled copy $U_\pi^\dagger OU_\pi$ of an observable $O$ has the same expectation value in $\overline\rho(t)$ as $O$.

In a permutation-invariant state $\rho$, such as $\overline\rho(t)$, a permutation-invariant observable therefore reduces to a few local ones.
For example, for qubits,
\begin{equation*}
	O=N\big(m_z+\tfrac12m_x^2\big)=\tfrac12+\sum_i\sigma^z_i+\frac1N\sum_{i<j}\sigma^x_i\sigma^x_j,
\end{equation*}
and since all $\sigma^z_i$, and all $\sigma^x_i\sigma^x_j$, have the same expectation value,
\begin{equation*}
	\tr(O\rho)=\tfrac12+N\tr(\sigma^z_1\rho)+\tfrac{N-1}2\tr(\sigma^x_1\sigma^x_2\rho).
\end{equation*}
The errors of the two local estimates are multiplied by $N$ and $(N-1)/2$, which is harmless as $\|O\|\ge N$ grows alike.
This needs the right split into pieces, though: $0=\sum_{i<j}(\sigma^z_i+\sigma^z_j)-(N-1)\sum_i\sigma^z_i$ splits the zero observable into two pieces whose estimates carry errors of order $N^2$ times the local accuracy.
We therefore use the canonical split into pieces that are traceless on each of their sites.

Let $\Gamma_0=I$, and let $\Gamma_1,\ldots,\Gamma_{\chi^2-1}$ be a basis of the traceless Hermitian operators on $\mathbb C^\chi$, e.g.\ the Pauli matrices.
Expanding $O$ in products of the $\Gamma_a$ and letting $B_S$ collect the products whose non-identity factors sit exactly on $S$ gives
\begin{equation}
	O=\sum_{j=0}^{k_O}\sum_{|S|=j}B_S,
	\label{eq:local-decomposition}
\end{equation}
where (i) $B_S$ acts on $S$ and its partial trace over any site of $S$ vanishes; (ii) $B_S=0$ for $|S|>k_O$, as $O$ is $k_O$-local; and (iii) $B_{\pi(S)}=U_\pi B_SU_\pi^\dagger$, as the expansion is unique and $O$ permutation invariant.
By (iii), all $B_S$ with $|S|=j$ are copies of one $j$-site operator $B^{(j)}$, itself invariant under permuting its sites; above, $B^{(0)}=\frac12$, $B^{(1)}=\sigma^z$ and $B^{(2)}=\sigma^x\otimes\sigma^x/N$.
As all copies have the same expectation value,
\begin{equation}
	\tr(O\rho)=\sum_{j=0}^{k_O}\binom Nj\tr\big(B^{(j)}_{\{1,\ldots,j\}}\rho\big),
	\label{eq:local-to-pi}
\end{equation}
with $B^{(j)}_{\{1,\ldots,j\}}$ acting on the first $j$ sites.
The $j=0$ term is the constant $\chi^{-N}\tr O$, and estimating the others to accuracy $\epsilon'\|B^{(j)}\|$ gives $\tr(O\rho)$ to accuracy $\epsilon'\sum_j\binom Nj\|B^{(j)}\|\le\epsilon'C_{\chi,k_O}\|O\|$ by \cref{lem:local-to-pi} below.

\begin{lemma}
	\label{lem:local-to-pi}
	There is a constant $C_{\chi,k_O}$, depending only on $\chi$ and $k_O$, such that $\sum_{j=0}^{k_O}\binom Nj\|B^{(j)}\|\le C_{\chi,k_O}\|O\|$ for every $N$.
\end{lemma}

\begin{proof}
	Since $|\tr(O\rho)|\le\|O\|$ for every state $\rho$, each state bounds $\|O\|$ from below, and identical product states are the natural choice here: the trace factorises over sites and all $\binom Nj$ copies of $B^{(j)}$ contribute equally.
	On these states, $O$ becomes a polynomial whose coefficients are the entries of the $\binom NjB^{(j)}$.

	In general, the $B_S$ do not depend on the choice of the $\Gamma_a$, so we take them orthonormal, $\tr(\Gamma_a\Gamma_b)=\delta_{ab}$, which gives $\|\Gamma_a\|\le1$; let $d=\chi^2-1$.
	By (i), $B^{(j)}=\sum_{a_1,\ldots,a_j}\beta^{(j)}_{a_1\cdots a_j}\Gamma_{a_1}\otimes\cdots\otimes\Gamma_{a_j}$ with real coefficients, which by (iii) are symmetric in their indices, and $\|B^{(j)}\|\le d^j\max|\beta^{(j)}|$.
	For $x\in\mathbb R^d$ with $|x|\le1/\chi$, $X=\sum_ax_a\Gamma_a$ has $\|X\|\le|x|\le1/\chi$, so $\sigma_x=I/\chi+X$ is a state.
	In $\tr(B_S\sigma_x^{\otimes N})$, each site outside $S$ contributes $\tr\sigma_x=1$, and on $S$ only $X$ survives by (i), so
	\begin{equation*}
		p(x)=\tr\big(O\sigma_x^{\otimes N}\big)=\sum_{j=0}^{k_O}\binom Nj\sum_{a_1,\ldots,a_j}\beta^{(j)}_{a_1\cdots a_j}x_{a_1}\cdots x_{a_j},
	\end{equation*}
	with $|p(x)|\le\|O\|$.
	As $\beta^{(j)}$ is symmetric, the coefficient of $x_{a_1}\cdots x_{a_j}$ is $\binom Nj\beta^{(j)}_{a_1\cdots a_j}$ times the number of distinct orderings of $(a_1,\ldots,a_j)$.
	It remains to bound the coefficients of $p$.
	A polynomial of degree at most $k_O$ in $d$ variables is determined by its values on a grid of $(k_O+1)^d$ points, here chosen inside the ball $|x|\le1/\chi$, and each of its coefficients is a fixed linear combination of these values, with weights depending only on $\chi$ and $k_O$.
	Every coefficient of $p$ is therefore at most $c\|O\|$, with $c$ depending only on $\chi$ and $k_O$, so $\binom Nj\|B^{(j)}\|\le d^jc\|O\|$, and summing over $j$ proves the claim.
\end{proof}

The $B^{(j)}$ are computed from $O$ in logarithmic space.
With the normalised partial trace $O|_T=\chi^{|T|-N}\tr_{\{1,\ldots,N\}\setminus T}O$, (i) gives $O|_T=\sum_{S\subseteq T}B_S$, since tracing out a site removes exactly the pieces acting on it.
Inclusion--exclusion inverts this to $B^{(j)}=\sum_{T\subseteq\{1,\ldots,j\}}(-1)^{j-|T|}O|_T\otimes I$.
For $O=Nf$ supplied as in \cref{eq:compression-normalized-family}, expanding each monomial over its site indices $i_1,\ldots,i_\ell$, the normalised partial trace of a term depends only on which indices coincide and which lie in $T$: each such group of index tuples contributes a product of normalised traces of single-site operator products, times the falling factorial of $N-|T|$ counting its tuples.
For $k_O,\chi=\mathcal O(1)$ there are $\mathcal O(1)$ groups, so every entry of $B^{(j)}$ is computed in polynomial time and $\mathcal O(\log N)$ workspace, and recomputed whenever the local algorithm reads it.

Hence, for $\chi,k_O=\mathcal O(1)$, the results on \cref{prob:averaged-dynamics} extend to $k_O$-local permutation-invariant observables: in each regime the relevant state, $\rho(t)$ for $J=0$ and $\overline\rho(t)$ otherwise, is permutation invariant, so estimating each $B^{(j)}_{\{1,\ldots,j\}}$ to accuracy $\epsilon\|B^{(j)}\|/C_{\chi,k_O}$ suffices.

\section{Universality of the Logical Gate Set}
\label{app:gate-universality}
The gates $A_i,B_i,C_{ij}$ in \cref{eq:hardness-gates}, together with their inverses, form an approximately universal gate set: they can approximate any unitary circuit to arbitrary accuracy, up to global phase. To see this, we express the standard universal gates $H_{\mathrm{Had}},T_{\mathrm{ph}},CZ$ as constant-length products of these gates. Suppressing the qubit index for single-qubit identities, Pauli multiplication gives
\begin{align}
	A^2B^2      & =-i(\mathsf X+\mathsf Z)/\sqrt2
	=-iH_{\mathrm{Had}},\nonumber                 \\
	ABA^\dagger & =e^{-i\pi\mathsf Z/8}
	=e^{-i\pi/8}T_{\mathrm{ph}},
\end{align}
where $H_{\mathrm{Had}}$ is the Hadamard gate and
$T_{\mathrm{ph}}=\operatorname{diag}(1,e^{i\pi/4})$.
For $F=B_i^2B_j^2$, the identity
$B^2\mathsf X B^{-2}=-\mathsf Z$ gives
\begin{align}
	FC_{ij}F^\dagger
	 & =e^{-i\pi\mathsf Z_i\mathsf Z_j/4},\nonumber \\
	CZ_{ij}
	 & =e^{-i\pi/4}
	e^{i\pi\mathsf Z_i/4}
	e^{i\pi\mathsf Z_j/4}
	e^{-i\pi\mathsf Z_i\mathsf Z_j/4}.
\end{align}
The local factors satisfy
$e^{i\pi\mathsf Z/4}=(ABA^\dagger)^{-2}$.
The controlled-$Z$ gate is $CZ=\operatorname{diag}(1,1,1,-1)$; conjugating it by a Hadamard on the target gives a controlled-NOT. Thus the construction contains the standard approximately universal Hadamard, $T$, and controlled-NOT gate set, up to global phase. Each replacement has constant length and can be generated sequentially in constant additional workspace. Such preprocessing has a known global phase; the direct compilation of an $A,B,C$ circuit has none.

\section{Proofs for \cref{sec:technical-discussion}}

\subsection{Proofs of \cref{sec:cluster-method}}
\label{app:cluster-count}
\begin{proof}[Proof of \cref{lem:clusters}]
The proof proceeds exactly as that of Proposition 3.6 of \citet{haah2021}: the count is maximized on the infinite $\Delta$-regular tree rooted at $v$, where it equals $\sum_{n=1}^{w}\binom{w-1}{n-1}D_n$, with $D_n=\frac{\Delta}{n(\Delta-1)+1}\binom{n(\Delta-1)+1}{n-1}$ the exact number of subtrees with $n$ nodes rooted at $v$ (Lemma 3.7 of \citet{haah2021}) and $\binom{w-1}{n-1}$ the number of assignments of multiplicities $\ge1$, summing to $w$, to the $n$ nodes.
The only departure is that we sharpen their rounded bound $D_n\le e\Delta(e(\Delta-1))^{n-1}$, the sole source of the prefactor $e\Delta$ in their Proposition 3.6, to
$D_n\le(e(\Delta-1))^{n-1}$,
so that the binomial theorem yields the claimed $\sum_{n=1}^{w}\binom{w-1}{n-1}D_n\le(1+e(\Delta-1))^{w-1}$; note also $1+e(\Delta-1)\le e\Delta$ since $1\le e$.
To prove the sharpened bound, note $D_1=1$; for $n\ge2$, write $\ell\coloneqq n-1$ and $K\coloneqq n(\Delta-1)+1$, so that $K-1=(\ell+1)(\Delta-1)$, and use the identity $\frac1K\binom K\ell=\frac1\ell\binom{K-1}{\ell-1}$:
\[
  \begin{aligned}
    D_n&=\frac{\Delta}{\ell}\binom{(\ell+1)(\Delta-1)}{\ell-1}\\
    &\le\frac{\Delta}{\ell}
    \frac{((\ell+1)(\Delta-1))^{\ell-1}}{(\ell-1)!}\\
    &=\frac{\Delta}{\Delta-1}
    \frac{(\ell+1)^{\ell-1}}{\ell!}(\Delta-1)^\ell
    \le(e(\Delta-1))^\ell,
  \end{aligned}
\]
where the last step combines $\frac{\Delta}{\Delta-1}\le2\le\ell+1$ with $(\ell+1)^\ell\le e^\ell\ell!$.
Finally, the claim for $\mathcal G^O_m$ follows by applying the above at $w=m+1$ with $v=O$ on the interaction graph of $H$ and $O$.
\end{proof}

\section{Numerics}
\label{sec:numerics}

\begin{figure*}[t]
  \centering
  \subfloat[\label{fig:hopfield-vs-sk-sz}]{%
    \includegraphics[width=0.604\linewidth]{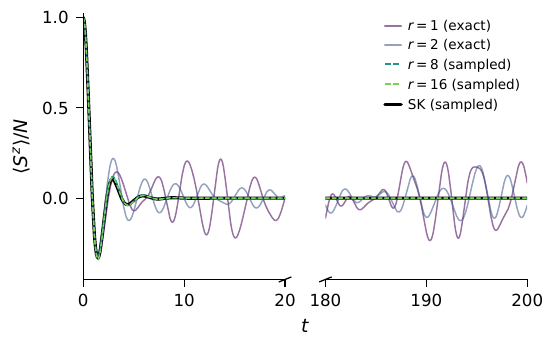}}%
  \hfill
  \subfloat[\label{fig:hopfield-vs-sk-dist}]{%
    \includegraphics[width=0.375\linewidth]{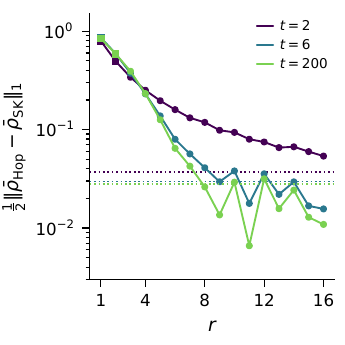}}
  \caption{The transverse-field Hopfield model $H_{\mathrm{Hop}}$ approaching SK with growing
  rank, at $N=16$ and $B=1.0$, with the cumulant-matched weights of \cref{eq:mu-cumulant}
  and the initial state $\ket{0}^{\otimes N}$.
  (a)~Disorder-averaged magnetization $\overline{\langle m_z(t)\rangle}$ on the windows
  $0\le t\le20$ and $180\le t\le200$, drawn on the same time scale. The thick black line is
  the SK model, \cref{eq:models-sk}, averaged over $200$ sampled Gaussian coupling
  matrices, with a band of $\pm2$ standard errors of the mean. Thin solid lines at $r=1,2$ are
  the exact pattern average from the symmetric-subspace simulation of the ancilla dilation
  (\cref{sec:numerical-methods}); dashed lines at $r=8,16$ are the average over $200$
  sampled pattern tables evolved in the full $2^N$-dimensional space, with their $\pm2$
  standard-error band. At $r=1,2$ the magnetization keeps oscillating with an amplitude of
  about $0.2$ up to $t=200$; at $r=8$ and $16$ the curves lie on top of SK within statistical
  error at all times shown.
  (b)~Trace distance $\tfrac12\|\overline{\rho}_{\mathrm{Hop}}-\overline{\rho}_{\mathrm{SK}}\|_1$
  between the disorder-averaged states of the two models as a function of the rank $r$, at the
  times $t=2$, $6$, and $200$ (colors). Every sampled state, Hopfield and SK, is averaged over
  site permutations before the disorder average, so that both states are evaluated in the
  Schur--Weyl block form of \cref{sec:numerical-methods}. Squares are the ranks with an exact pattern average
  ($r=1,2$), so that their distance to SK carries the sampling noise of the SK average only;
  circles are the sampled ranks, which carry the noise of both averages. The dotted horizontal
  line of each color is the sampling floor at that time: the trace distance that finite
  sampling alone produces, estimated by splitting the realizations of each model into two
  halves and adding the two split-half distances in quadrature, taken as the median over the
  sampled ranks. A point on or below its floor is indistinguishable from SK at this number of
  realizations. The distance at $t=200$ reaches the floor by $r\approx8$, whereas at $t=2$ it
  decreases more slowly and stays resolved up to $r=16$. Odd ranks lie above the even ones
  because the triangle correlation of \cref{eq:hopfield-triangle} cancels only for even
  $r$.}
  \label{fig:hopfield-vs-sk}
\end{figure*}

We now numerically consider ensemble averaged dynamics in the setting of \cref{sec:summary-hopfield}, i.e. $r=\mathcal{O}(1)$, and more specifically the quantum Hopfield models of \cref{eq:models-hopfield}.

First, we discuss how, for a given, fixed $r$, the parameters $\mu_\ell$ can be chosen so that the low-order statistics of the Hopfield bonds match those of SK as closely as the rank allows.
As discussed in \cref{sec:reduction-fixed-rank}, their discrete pattern disorder admits the permutation-invariant dilation, and thus a reduction of \cref{prob:averaged-dynamics} to the permutation-invariant dynamics estimation of \cref{eq:task-fixed}; utilising this, we discuss some applicable numerical techniques in \cref{sec:numerical-methods}.
Finally, we discuss some results in \cref{sec:numerical-results}.

\subsection{Choice of the Hopfield Weights}
\label{sec:hopfield-weights}

A useful starting point for connecting the Hopfield and SK models is to consider a low-rank approximation of the SK interaction matrix in \cref{eq:models-sk}. Retaining a few terms in its spectral decomposition expresses the interaction as a weighted sum of squared collective spin operators, one for each retained eigenvector. The Hopfield model in \cref{eq:models-hopfield} has this same structure, but uses independent random sign patterns $v_{i\ell}$ rather than SK eigenvectors.
It is therefore not a spectral truncation of an SK realization.
However, at fixed even rank $r$, the weights $\mu_\ell$ can be chosen to match the lowest-order bond moments of SK. We first match the variance of each pair coupling, then examine loop correlations among bonds.

Expanding the squares in \cref{eq:models-hopfield} gives a constant term $(\sum_\ell\mu_\ell)I$, since $(\sigma_i^z)^2=I$, and the pair couplings
\begin{equation}
  K_{ij}=\frac2N\sum_{\ell=1}^r\mu_\ell v_{i\ell}v_{j\ell},
  \label{eq:hopfield-pairs}
\end{equation}
where the factor of two counts the two orders of each pair. The weights $\mu_\ell$ are fixed and real. We use an overbar for the average over the independent uniform signs $v_{i\ell}$, as in \cref{eq:task-averaged}. For $i\ne j$, each product $v_{i\ell}v_{j\ell}$ has zero mean and unit variance, and contributions from different patterns are independent. Thus
$\overline{K_{ij}}=0$ and $\overline{K_{ij}^2}=4\sum_\ell\mu_\ell^2/N^2$.
Matching the SK bond variance $1/N$ requires $\sum_\ell\mu_\ell^2=N/4$.
Still, SK bonds are independent, whereas Hopfield bonds share the same patterns. The simplest way to see the difference is to multiply the bonds around a triangle. For distinct sites $i,j,k$,
\begin{equation}
  \overline{K_{ij}K_{jk}K_{ki}}
    =\frac{8}{N^3}\sum_\ell\mu_\ell^3.
  \label{eq:hopfield-triangle}
\end{equation}
Indeed, a product of independent signs has a nonzero average only if every sign occurs an even number of times. In this triangle, that requires all three bonds to contribute the same pattern index. The corresponding average is zero for independent SK bonds. We can remove this difference by pairing equal and opposite weights, so that $\sum_\ell\mu_\ell^3=0$.
Next consider a loop through four distinct sites. The same sign-counting argument gives
$
  \overline{K_{ij}K_{jk}K_{kl}K_{li}}
    =\frac{16}{N^4}\sum_\ell\mu_\ell^4.
$
This average also vanishes in SK, but cannot vanish here: all fourth powers are nonnegative, and the variance constraint requires nonzero weights. We therefore make it as small as possible. At fixed $r$ and $\sum_\ell\mu_\ell^2=N/4$, the Cauchy--Schwarz inequality gives
\begin{equation}
  \sum_\ell\mu_\ell^4\geq\frac1r
  \left(\sum_\ell\mu_\ell^2\right)^2
  =\frac{N^2}{16r},
\end{equation}
with equality precisely when all weights have equal magnitude. For even $r$, equal magnitudes and equal numbers of positive and negative weights satisfy both requirements. The variance constraint then fixes their magnitude, giving the choice
\begin{equation}
  \mu_\ell=(-1)^{\ell-1}\frac{\sqrt N}{2\sqrt r}.
  \label{eq:mu-cumulant}
\end{equation}
Thus equal magnitudes minimize the four-bond loop correlation, while balanced signs cancel the triangle correlation, without changing the matched SK variance. The remaining four-bond average is $1/(rN^2)$, or $1/r$ for bonds normalized to unit variance. For odd $r$, equal magnitudes still minimize this average, but cannot also cancel the triangle correlation.

The same choice also gives a Gaussian limit as the number of patterns grows. Each normalized bond $\sqrt N K_{ij}$ is a sum of $r$ independent, mean-zero pattern contributions of magnitude $1/\sqrt r$. For any fixed collection of distinct bonds, the contribution vectors are independent across patterns and have zero cross-covariances between bonds. Their bounded contributions satisfy the multivariate central limit theorem, so the normalized bonds converge jointly to independent standard Gaussians as $r\to\infty$. At finite rank, the loop correlations above quantify a departure from this limit.

This optimality concerns the bond variance and the two loop correlations, not the full dynamics. At fixed $r$, a configuration aligned with a pattern receives an energy per site of magnitude $|\mu_\ell|\propto\sqrt{N/r}$ from that pattern, which can grow with $N$. Moreover, convergence of a fixed collection of bonds does not give an error bound uniform in system size or evolution time. The finite-pattern model is therefore a statistically matched low-rank analogue, not a controlled approximation to SK dynamics.

\subsection{Numerical Techniques}
\label{sec:numerical-methods}

\paragraph{Exact symmetric subspace simulation}
For a permutation-invariant Hamiltonian and an initial pure state in the symmetric subspace, the dynamics can be propagated directly in the compressed form of \cref{sec:compressed-encoding}: we evolve $\varrho(t)=e^{-iht}\varrho(0)e^{iht}$ in the occupation basis $\ket{\mathbf n}_N$ and evaluate $F_O(t)=\tr(o\varrho(t))$ with $o=V^\dagger OV$ as in \cref{eq:compression-compressed-state}. Numerically, only the valid occupations need to be stored, of which there are
\begin{equation}
  D_{N,\chi}=\binom{N+\chi-1}{\chi-1}=\mathcal O(N^{\chi-1})\le Q
  \label{eq:numerical-symmetric-dimension}
\end{equation}
for fixed $\chi$, instead of $\chi^N$.
The matrix elements of $h$ follow from \cref{eq:compression-matrix-units}.
Each factor $m(B)$ in \cref{eq:compression-matrix-units} moves at most one site from one state to another, so $h$ connects $\ket{\mathbf n}_N$ only to occupations reached by at most $k$ such moves, of which there are at most $\chi^{2k}$; each row of $h$ thus has only a constant number of nonzero entries at fixed locality and $\chi$.

The resulting sparse $h$ can be diagonalized at small $D_{N,\chi}$ or applied through a Krylov approximation to the exponential at larger sizes, without a many-body truncation; numerical propagation errors must still be controlled. State-vector storage is $\mathcal O(D_{N,\chi})$, not logarithmic in $N$.

The same method applies to the Hopfield dilation of \cref{sec:reduction-fixed-rank}, where the enlarged sites have $\chi=2^{r+1}$. For $\rho(0)=(\ket{\phi}\bra{\phi})^{\otimes N}$, the enlarged initial state $(\ket{\phi}\otimes\ket{+}^{\otimes r})^{\otimes N}$ is a symmetric pure state, as required for \cref{eq:compression-compressed-state}, and evaluating $O\otimes I$ gives the exact discrete-disorder average, up to numerical error, without sampling patterns. The dimension in \cref{eq:numerical-symmetric-dimension} grows rapidly with $r$, limiting accessible ranks.

\paragraph{ED \& sampling}
To access larger ranks, we sample independent pattern tables with fair signs $v_{i\ell}\in\{-1,+1\}$ and construct the couplings in \cref{eq:hopfield-pairs}, keeping the weights fixed. For each realization, we evolve $\ket{0}^{\otimes N}$ in the full $2^N$-dimensional Hilbert space via standard ED.
Averaging the resulting expectation values estimates the disorder-averaged dynamics, naturally with a time complexity exponential in $N$.
The SK averages are sampled in the same way, with independent Gaussian couplings.
For trace distances we average the states themselves. An average of $M$ pure states has rank at most $M$, and hence trace distance at least $1-M/2^N$ from $I/2^N$, whatever the dynamics.
However, the disorder ensembles are invariant under relabelling the sites, and so is $\overline\rho(t)$ - we therefore average each sampled state over all site permutations before averaging over realizations, which leaves the estimate unbiased.
By Schur--Weyl duality the result has the block form $\bigoplus_J\rho_J\otimes I_{m_J}$, where $\rho_J$ acts on the spin-$J$ irreducible representation and $m_J$ is its multiplicity, so that $\sum_J(2J+1)^2=\mathcal O(N^3)$ numbers fix it; all trace distances are evaluated in this form.
The exact ancilla averages are permutation invariant by construction and have the same form without this step.

\subsection{Numerical Results}
\label{sec:numerical-results}

\Cref{fig:hopfield-vs-sk} compares the disorder-averaged Hopfield with the weights of \cref{eq:mu-cumulant} and SK dynamics; as described in the corresponding caption at growing ranks they recover the behaviour of the SK model.

\section{Random Transverse-Field Dynamics and the LMG Reduction}
\label{app:random-transverse}
\subsection{Random Transverse Signs and Exact Reduction}
\label{sec:random-transverse}

\begin{figure*}[t]
  \centering
  \includegraphics[width=\textwidth]{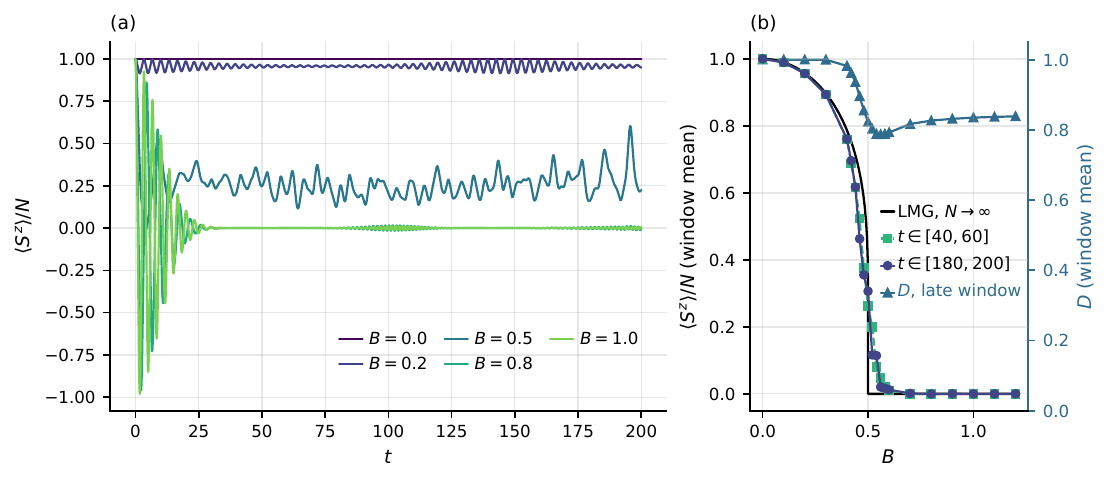}
  \caption{Random transverse-field dynamics for $N=100$, initially polarized in $z$, with $J=1$.
  (a) Longitudinal magnetization $\langle m_z\rangle$ for $B=0,0.2,0.5,0.8,1$, sampled at intervals $0.1$ through $t=200$.
  (b) Signed magnetization averaged over $t\in[40,60]$ and $[180,200]$ (left axis), compared with the classical LMG orbit mean (black curve). Each quantum mean contains 201 samples including both endpoints.
  The triangle-marked curve gives the mean trace distance $D(t)=\tfrac12\|\rho_{\mathrm{phys}}(t)-I/2^N\|_1$ over $[180,200]$ (right axis).
  The classical curve takes the infinite-size limit before the infinite-time limit; the quantum curves are finite-window readouts, not stationary-state values.}
  \label{fig:rt-results}
\end{figure*}

Another model fitting the disorder-reduction framework of \cref{sec:reduction-fixed-rank} is
\begin{align}
	H_s & =\frac{1}{N}\sum_{i<j}\sigma_i^z\sigma_j^z
	+B\sum_i s_i\sigma_i^x,
	    & s_i                                        & \overset{\mathrm{i.i.d.}}{\sim}\operatorname{Unif}\{-1,+1\}.
	\label{eq:rt-hamiltonian}
\end{align}

Following \cref{sec:reduction-fixed-rank}, introduce an ancilla Pauli operator $\tau_i^z$ at each site and define
\begin{align}
	\bar H       & =\frac{1}{N}\sum_{i<j}\sigma_i^z\sigma_j^z
	+B\sum_i\sigma_i^x\tau_i^z,                                                                  \\
	\ket{\Psi_0} & =(\ket{0}\otimes\ket{+})^{\otimes N},
	             & \ket{+}                                    & =\frac{\ket{0}+\ket{1}}{\sqrt2}.
\end{align}
The conserved ancilla eigenvalues label the realizations. Writing
$\tau_i^z\ket{s}_a=s_i\ket{s}_a$, orthogonality gives
\begin{align}
	\ket{\Psi(t)}
	 & =2^{-N/2}\sum_s e^{-iH_st}\ket{0^N}\otimes\ket{s}_a, \\
	\rho_{\mathrm{phys}}(t)
	 & =\tr_a\ket{\Psi(t)}\bra{\Psi(t)}
	=2^{-N}\sum_s e^{-iH_st}\rho_0e^{iH_st}.
	\label{eq:rt-ancilla-trace}
\end{align}
Thus measuring $O\otimes I$ in the enlarged system gives the disorder-averaged expectation of any fixed physical observable $O$, mapping \cref{prob:averaged-dynamics} to the permutation-invariant dynamics estimation of \cref{eq:task-fixed}.
The Hamiltonian and initial state are symmetric under permutations of the enlarged sites, each of local dimension $\chi=4$. Their symmetric occupation basis has dimension $\binom{N+3}{3}$, counting four nonnegative occupations summing to $N$.

Similarly to \cref{app:rank-one-gauge}, the signs can be removed by a local gauge transformation, relating the model to clean LMG dynamics.
Let $H_+$ be the uniform-field Hamiltonian and set
\begin{align}
	b_i & =(1-s_i)/2,
	    & G_s         & =\prod_i(\sigma_i^z)^{b_i}.
\end{align}
Since $\sigma^z\sigma^x\sigma^z=-\sigma^x$ and
$(\sigma^z)^2=I$,
\begin{align}
	H_s                 & =G_sH_+G_s,
	                    & G_s\ket{0^N}       & =\ket{0^N},           \\
	e^{-iH_st}\ket{0^N} & =G_s\ket{\phi(t)},
	                    & \ket{\phi(t)}      & =e^{-iH_+t}\ket{0^N}.
\end{align}
Define computational-basis dephasing by the trace-preserving channel
$\Delta_z(X)=\sum_u\ket{u}\bra{u}X\ket{u}\bra{u}$, where
$u\in\{0,1\}^N$. For two bit strings $u,v$, independent uniform signs give
\begin{align}
	\bra{u}\left(2^{-N}\sum_sG_sXG_s\right)\ket{v}
	 & =X_{uv}\prod_i\frac{1+(-1)^{u_i+v_i}}{2}\nonumber \\
	 & =X_{uv}\delta_{u,v}.
\end{align}
Consequently, the exact reduced state is
\begin{align}
	\rho_{\mathrm{phys}}(t)
	       & =\Delta_z\!\left(\ket{\phi(t)}\bra{\phi(t)}\right)
	=\sum_u p_u(t)\ket{u}\bra{u},                               \\
	p_u(t) & =\big|\bra{u}\ket{\phi(t)}\big|^2.
	\label{eq:rt-dephased-state}
\end{align}
This identity determines every single-time physical expectation:
\begin{align}
	\tr(O\rho_{\mathrm{phys}}(t))
	 & =\sum_u p_u(t)\bra{u}O\ket{u}\nonumber  \\
	 & =\bra{\phi(t)}\Delta_z(O)\ket{\phi(t)}.
	\label{eq:rt-observables}
\end{align}
Applying the same formula to measurement effects determines all single-time outcome probabilities. Diagonal observables, including longitudinal correlations and the full magnetization distribution, agree with clean LMG dynamics in every realization because they commute with $G_s$. Transverse one-spin expectations instead satisfy
$\langle\sigma_i^a\rangle_s=s_i\langle\sigma_i^a\rangle_+$ for
$a\in\{x,y\}$ and vanish only after the sign average. Their realization-dependent values need not vanish.

\subsection{Classical Reference and Finite-Window Dynamics}
\label{sec:rt-classical}

To establish the classical reference with the present Pauli normalization, use
$Nm_z^2=I+\frac2N\sum_{i<j}\sigma_i^z\sigma_j^z$, so that
$H_+=N(m_z^2/2+Bm_x)-I/2$.
The commutators $[m_a,m_b]=\frac{2i}{N}\epsilon_{abc}m_c$ give
\begin{align}
	\dot m_x & =-\{m_z,m_y\},\nonumber           \\
	\dot m_y & =\{m_z,m_x\}-2Bm_z,
	         & \dot m_z                  & =2Bm_y.
\end{align}
Here $\{A,C\}=AC+CA$. In the mean-field limit from the polarized coherent state, factorizing quadratic expectations yields
\begin{align}
	\dot x & =-2yz,
	       & \dot y & =2z(x-B),
	       & \dot z & =2By,
	\label{eq:rt-classical-eom}
\end{align}
where $(x,y,z)=\lim_{N\to\infty}\langle(m_x,m_y,m_z)\rangle$
starts at $(0,0,1)$. These equations conserve
$x^2+y^2+z^2=1$ and $e=z^2/2+Bx=1/2$.
For $B>0$, eliminating $x$ and $y$ gives
\begin{align}
	x        & =\frac{1-z^2}{2B},
	         & y^2                   & =1-z^2-\frac{(1-z^2)^2}{4B^2},\nonumber \\
	\dot z^2 & =(1-z^2)(z^2+4B^2-1).
	\label{eq:rt-orbit}
\end{align}
The fixed point $(1,0,0)$ has energy $e=B$. Linearizing about it gives
$\ddot z=4B(1-B)z$, making it a saddle for $0<B<1$.
The north-pole orbit meets its separatrix when $B=1/2$.
Below this field the initial energy is above the saddle energy, and
$z$ remains between $\sqrt{1-4B^2}$ and $1$.

For $0<B<1/2$, set $q=4B^2$ and parameterize the descending half-orbit by
$z=\sqrt{1-q\sin^2\theta}$, $0\le\theta\le\pi/2$.
Then $dz=-q\sin\theta\cos\theta\,d\theta/z$ and
$\sqrt{(1-z^2)(z^2+q-1)}=q\sin\theta\cos\theta$, so $dt=d\theta/z$.
The two half-orbits therefore give
\begin{align}
	T & =2K(q), \qquad
	K(q)=\int_0^{\pi/2}\frac{d\theta}{\sqrt{1-q\sin^2\theta}}, \\
	\int_0^T z(t)\,dt & =2\int_0^{\pi/2}d\theta=\pi.
\end{align}
Here $K$ uses the elliptic \emph{parameter} $q$, not the modulus.
For $B>1/2$, the orbit crosses both hemispheres; its speed in
\cref{eq:rt-orbit} is even in $z$, so its signed magnetization integral vanishes.
At $B=1/2$, $z(t)=\operatorname{sech}t$ solves the descending equation and has zero infinite-time mean. At $B=0$, $z(t)=1$.
Combining these cases gives
\begin{align}
	\overline m_z^{\mathrm{cl}}
	 & \coloneqq\lim_{\mathcal T\to\infty}
	\frac{1}{\mathcal T}\int_0^{\mathcal T}z(t)\,dt\nonumber \\
	 & =\begin{cases}
		\pi/[2K(4B^2)], & 0\le B<1/2, \\
		0,              & B\ge1/2.
	\end{cases}
	\label{eq:rt-orbit-mean}
\end{align}
The endpoint singularity of $K(q)$ is proportional to the integral of
$[(1-q)+(\pi/2-\theta)^2]^{-1/2}$. Its logarithmic divergence makes the orbit mean approach zero continuously, not through a finite jump.

\Cref{fig:rt-results} shows exact numerics via the reduction of \cref{sec:reduction-fixed-rank} as well as a comparison of the numerically determined plateau values with the LMG discussion of this section.

\bibliographystyle{apsrev4-2}
\bibliography{references}

\end{document}